\documentclass[11pt]{article} 

\usepackage[utf8]{inputenc} 

\usepackage{geometry} 
\usepackage{graphicx} 

\usepackage{booktabs} 
\usepackage{array} 
\usepackage{paralist} 
\usepackage{verbatim} 
\usepackage{subfig} 
\usepackage{amsmath}
\usepackage{amsfonts}
\usepackage{amssymb}
\usepackage{mathrsfs}
\usepackage{color}
\usepackage{amsthm}
\usepackage{float}
\usepackage{tikz}
\usepackage{fancyhdr} 
\usepackage{sectsty}
\allsectionsfont{\sffamily\mdseries\upshape} 

\usepackage[nottoc,notlof,notlot]{tocbibind} 
\usepackage[titles,subfigure]{tocloft} 

\newcommand{\C}{\mathbb{C}}

\newcommand{\N}{\mathbb{N}}

\newcommand{\Hil}{\mathcal{H}}
\newcommand{\A}{\mathscr{A}}

\newcommand{\leftset}{\left\{\left.}
\newcommand{\midsetl}{\right\vert\left.}
\newcommand{\midsetr}{\right.\left\vert}
\newcommand{\rightset}{\right.\right\}}
\newcommand{\style}{\mathcal}
\newcommand{\tr}{\text{Tr}}
\newcommand{\Tr}{\tr}
\newcommand{\UCP}{\Ep}
\newcommand{\Span}{\rm{span}}
\newcommand{\alg}{\rm{alg}}
\newcommand{\B}{\mathscr{B}}

\newcommand{\Ep}{\mathcal{E}}
\newcommand{\Md}{\mathcal{M}_{d}}
\newcommand{\Me}{\mathcal{M}_{\mathcal{E}}}
\newcommand{\M}{\mathcal{M}}
\newcommand\numberthis{\addtocounter{equation}{1}\tag{\theequation}}
\newtheorem{definition}{Definition}[section]
\newtheorem{theorem}[definition]{Theorem}
\newtheorem{lemma}[definition]{Lemma}
\newtheorem{corollary}[definition]{Corollary}

\theoremstyle{definition}
\newtheorem{remark}[definition]{Remark}

\newtheorem{cor}{Corollary}[section]
\newtheorem{prop}{Proposition}[section]

\providecommand{\keywords}[1]{\textit{Key words:} #1}

\title{Spectral Properties of Tensor Products of Channels}
\author{Sam Jaques and Mizanur Rahaman}
\AtEndDocument{{Samuel Jaques, Department of Combinatorics and Optimization, University of Waterloo, Waterloo, ON, Canada N2L 3G1
\newline
\textbf{e-mail:}  sejaques@uwaterloo.ca
}\vspace{1cm}

{Mizanur Rahaman, Department of Pure Mathematics and Institute for Quantum Computing, University of Waterloo,
Waterloo, ON, Canada N2L 3G1
\newline
\textbf{e-mail:} mizanur.rahaman@uwaterloo.ca
}}
 
\begin{document}
\maketitle
\begin{abstract}
We investigate spectral properties of the tensor products of two completely positive and trace preserving linear maps (also known as quantum channels) acting on matrix algebras. This leads to an important question of when an arbitrary subalgebra can split into the tensor product of two subalgebras. We show that for two unital quantum channels the multiplicative domain of their tensor product splits into the tensor product of the individual multiplicative domains. Consequently, we fully describe the fixed points and peripheral eigen operators of the tensor product of channels. Through a structure theorem of maximal unital proper *-subalgebras (MUPSA) of 
a matrix algebra we provide a non-trivial upper bound of the recently-introduced multiplicative index of a unital channel. This bound gives a criteria on when a channel cannot be factored    
into a product of two different channels. We construct examples of channels which cannot be realized as a tensor product of two channels in any way. With these techniques and results, we found some applications in quantum information theory.
\end{abstract}
\footnote{
}
\section{Introduction}
If we have a linear map acting on a matrix algebra that can be expressed as a tensor product of matrix algebras, and the map itself can be expressed as a tensor product of two other linear maps, there may be few similarities between the constituent maps and the larger linear map they produce. If we restrict ourselves to special classes of linear maps and special domains of matrix algebras, then the tensor product adds no extra complexity. Our goal in this paper is to use the multiplicative domain to characterize some of these properties for trace-preserving, completely positive maps on matrix algebras. These maps are also known as quantum channels, which we refer to as channels.\par
The multiplicative domain of a linear map  $\UCP:\M_d\rightarrow \M_d$ is the set of all matrices $x\in \M_d$ such that, for all $y\in \M_d$, $\UCP(xy)=\UCP(x)\UCP(y)$ and $\UCP(yx)=\UCP(y)\UCP(x)$. 
When the linear map is completely positive, this specific set has received much attention in  operator theory and operator algebras (\cite{choi1}, \cite{paulsen}-chapter 4, \cite{stormer}-section 2.1). In this context, it characterizes certain distinguishability measures. A completely positive linear map acts like a homomorphism on the multiplicative domain, and hence studying this domain can reveal structure and properties of the linear map.\par
In quantum information theory (\cite{johnston}, \cite{choi2}, \cite{kribs-spekkens}), the multiplicative domain contains the unitarily correctable codes and noiseless subsystems. Studying the multiplicative domain of tensor products sheds light on error correction in bipartite systems.\par
It turns out that we can capture most of the spectral properties of the tensor product of channels simply by investigating the multiplicative behavior. Note that the spectral properties of a channel acting on one copy of a quantum system have been well explored 
(\cite{wolf}, \cite{ergodic}, \cite{evans-krohn}, \cite{inverse-eigenvalue}) for various purposes, mainly in an effort to understand the dynamics of a system evolving through  quantum measurements. In quantum dynamical systems, the ergodicity of a channel \cite{ergodic} and its decoherence-free subspaces \cite{veronica} are important spectral properties . When the underlying domain is a bipartite system, the spectral properties of product channels can be hard to analyze, but we can use the multiplicative domain as a tool to understand them.\par
As in previous work on quantum error correction (e.g., \cite{johnston}\cite{kribs-spekkens}), we restrict our focus to unital channels because the multiplicative domain has less structure in non-unital channels. In particular, the multiplicative domain of a unital channel can be described using the Kraus operators. Using this, we can characterize certain channels and derive facts beyond the multiplicative structure. 
\par
The paper is organized as follows:
firstly, in Section \ref{sec:mutiplicatve_domain_of_products}, we show that the multiplicative domain of a tensor product of unital channels ``splits" nicely with the tensor product. We use this to prove that the peripheral spectra of two unital channels will precisely determine whether the fixed points of their tensor product will also split or not.
This analysis provides the necessary and sufficient condition on when the tensor product of two ergodic (or primitive) channels is again ergodic (or primitive).
Here we recapture some of the results obtained by 
\cite{Luczak}, \cite{Watanabe} in a very different way based on the analysis of multiplicative domain. 
\par
Since \cite{miza} showed that repeated applications of a finite-dimensional channel produces a chain in the lattice of unital *-subalgebras of $\M_d$, we characterize such algebras in Section \ref{sec:MUPSAS}. This provides an easy way to enumerate the lattice of unital *-subalgebras of $\M_d$, as well as providing a limit on the length of chains in the lattice that is linear in the dimension. This finding can be of independent interest because it provides a finer analysis of the structure of unital *-subalgebras in $\M_d$.  
In turn, this allows us to use the multiplicative index, introduced in \cite{miza}, to show that certain channels cannot be product channels. We give examples of channels with large multiplicative indices in Sections \ref{sec:ETB_channels} and \ref{sec:Schur_channels}, thus showing that these cannot be product channels.\par
Next, in Section \ref{Sec:strictly-contractive}, we consider channels which are strictly contractive with respect to some distinguishability measures that frequently arise in information theory. Using our results in the previous sections we prove that the tensor product of two strictly contractive channels with respect to certain distinguishability
measures is again strictly contractive provided the measures allow recovery maps. We make use of the reversibility and monotonicity properties of these measures under channels, which is a wide topic of current research (\cite{Jencova1}, \cite{Jencova2}, \cite{f-div1}, \cite{f-div2}).   
\par
As a final application, we show that unitary-correctable quantum codes (UCC) gain nothing through tensor products. 

\subsection{Background and Notation}
Throughout this paper we will use the following notation:
\begin{itemize}
	\item
	$\UCP,\Psi$ will refer to quantum channels, that is, completely positive, trace-preserving linear operators from $B(\Hil)$ to $B(\Hil)$ for some finite dimensional Hilbert space $\Hil$. In this paper we  identify $B(\Hil)$ with $\M_d$, the $d\times d$ complex matrices. It is well known that a quantum channel $\UCP:\M_d\rightarrow\M_d$ is always represented by a set of (non-unique) Kraus operators $\{a_{j}\}_{j=1}^{n}$ in $\M_d$ such that for all $x\in \M_d$, we have 
	\[\UCP(x)=\displaystyle\sum_{j=1}^{n}a_{j}xa_{j}^*,\]
where $\displaystyle\sum_{j=1}^{n}a_{j}^*a_{j}=1$. Here $1$ represents the identity matrix in $\M_d$. In the dual picture $\UCP$ is realized as a unital completely positive map and denoted by $\UCP^*$ acting on $\M_d$ again and satisfying the relation
\[\Tr(\UCP(x)y)=\Tr(x\UCP^*(y)),\]
for every $x,y\in \M_d$. This linear map $\UCP^*$ is the adjoint of $\UCP$ with respect to the Hilbert-Schmidt inner product which is defined as $\langle a,b\rangle_{HS}=\tr(ab^*)$, for all $a,b\in \M_d$.\par
An important note is that many papers work in the dual framework, where a quantum channel is necessarily unital but may not be trace-preserving. Hence, these papers refer to these maps as unital channels, or UCP maps. In our work, where all channels are trace-preserving, unitality is an extra condition that limits our results to a particular subset of quantum channels.
\item
Lowercase letters from the end of the Latin alphabets, $x,y,z$, will refer to matrices in $\M_d$. The letters $p,q$ will refer to projections in $\M_d$.
\item
Greek letters, $\varphi,\zeta$ will refer to either vectors in $\C^d$ or partitions of $\{1,\cdots,n\}$ for $n\in\N$. We will use $[n]$ to denote the set $\{1,\cdots,n\}$.
\item
Stylized letters from the beginning of the Latin alphabet, $\A,\B,\mathscr{C}$, will refer to sub-algebras of $\M_d$. For a set $S$, the algebra generated by $S$ will be denoted $\text{alg}(S)$ and the *-algebra generated by $S$ will be denoted $\text{alg}^*(S)$.
\item
For a quantum channel $\UCP$, $\M_\UCP$ denotes the multiplicative domain and also 
$\rm{Fix}_{\UCP}$ denotes the set of fixed points of $\UCP$, that is, 
\[\rm{Fix}_{\UCP}=\{a\in \M_d \ | \ \UCP(a)=a\}.\]  
\end{itemize}

There are a number of useful characterizations of the multiplicative domain we will use extensively.

\begin{theorem}[See \cite{choi1}]\label{choi}
	For a unital completely positive map $\UCP:\M_d\rightarrow \M_d$, the multiplicative domain $\M_\UCP$ is a C$^*$-subalgebra of $\M_d$ and moreover, it is equal to the following set:
	\[\M_\UCP=\{x\in \M_d \ | \ \UCP(x^*x)=\UCP(x)^*\UCP(x), \ \UCP(xx^*)=\UCP(x)\UCP(x^*)\}.\]
\end{theorem}
The following theorem is also useful in describing the multiplicative domain of a unital channel.
\begin{theorem}[See \cite{kribs-spekkens},\cite{miza}]\label{kribs-spekkens}
For a unital channel $\UCP$, we have the relation 
\[\M_\UCP=\rm{Fix}_{\UCP^*\circ\UCP}.\]
Here $\UCP^*$ is the adjoint of $\UCP$
\end{theorem} 
The next theorem connects the fixed points set and the Kraus operators of a channel.
\begin{theorem}[See \cite{krbs}]\label{cummutant-fixed pnt}
Let $\UCP$ be a unital channel represented as $\UCP(x)=\displaystyle\sum_{j=1}^{n}a_{j}xa_{j}^*$. Then the fixed point set $\rm{Fix}_\UCP$ is an algebra and it equals to the commutant of the *-algebra ($\A$) generated by $\{a_1,\cdots, a_{n}\}$. That is 
\[\rm{Fix}_\UCP=\A',\]
where $\A'$ represents the commutant of the algebra $\A$. 
\end{theorem}
It follows that $\M_\UCP$ is a *-closed subalgebra of $\M_d$ containing the fixed points of $\UCP$. As with all finite *-algebras, it is generated by a set of projections. For any projection $p\in \M_\UCP$, $\UCP(p)$ is a projection of the same rank, and $1-p$ is also in $\M_\UCP$. We say that $\M_\UCP$ is \textit{trivial} if $\M_\UCP=\C 1$; if $\M_\UCP$ is non-trivial, then it must contain at least two orthogonal projections.\par
For any unital channel $\UCP$ and any $k\in\N$, $\M_{\UCP^{k+1}}\subseteq \M_{\UCP^k}$ \cite{miza}, and hence there is some $N\in\N$ such that for any $n\geq N$, $\M_{\UCP^n}=\M_{\UCP^N}$. Following \cite{miza}, we denote this algebra $\M_{\UCP^\infty}$ and refer to it as the \textit{stabilized multiplicative domain} of $\UCP$.
\begin{definition}[\cite{miza}]
	The \textit{multiplicative index} of a unital quantum channel $\UCP$ is the minimum $n\in\N$ such that $\M_{\UCP^n}=\M_{\UCP^\infty}$.
\end{definition}
We denote the multiplicative index of $\UCP$ by $\kappa(\UCP)$.
Another useful result is Lemma 2.2 from \cite{miza}:
\begin{lemma}\label{lem:md_composition}
	If $\UCP_1,\UCP_2$ are two unital quantum channels, then
	\[\M_{\UCP_2\circ\UCP_1}=\{x\in \M_{\UCP_1} \ | \ \UCP_1(x)\in \M_{\UCP_2}\}.\]
\end{lemma}


\section{The Multiplicative Domain of Product Channels}\label{sec:mutiplicatve_domain_of_products}
\subsection{Splitting problem for subalgebras in tensor product}
The splitting problem for a von Neumann subalgebra (or a C$^*$-subalgebra) of a tensor product of algebras has remained one of the most important problems in operator algebra. One of the early results that drew a lot of attention on this problem is due to L. Ge and R. Kadison:
\begin{theorem}\rm{(Ge-Kadison, 1996, \cite{Ge-Kadison})}
Let $\mathcal{M}, \mathcal{N}$ be two von Neumann algebras and assume that $\mathcal{M}$ is a factor.
If $\mathcal{A}\subseteq \mathcal{M}\bar{\otimes}\mathcal{N}$ is a subalgebra that contains $\mathcal{M}\otimes \mathbb{C}1$, then 
\[\mathcal{A}=\mathcal{M}\otimes \mathcal{B},\]
for some von Neumann subalgebra $\mathcal{B}$ of $\mathcal{N}$.
\end{theorem}
There have been a lot of improvements and new research into the splitting problem. See \cite{splitting1}, \cite{splitting2},\cite{splitting3},\cite{splitting4}
for more information on this area.

Here we examine the multiplicative domain of the tensor product of two channels $\UCP_1,\UCP_2$ acting on $\mathcal{M}_{d}$ and $\mathcal{M}_{d'}$ separately. Since the multiplicative domain $\mathcal{M}_{\UCP_1\otimes\UCP_2}$ is a C$^*$-subalgebra of $\mathcal{M}_{d}\otimes\mathcal{M}_{d'}$, it is natural to ask whether this subalgebra splits into tensor product of two subalgebras. We show that for unital channels, the multiplicative domain is unchanged by the tensor product. To prove this claim we need the following lemma. For our purposes, if $\A_1$ and $\A_2$ are two algebras and $S\subseteq\A_1$ and $R\subseteq\A_2$ are two sets, possibly without any algebraic structure themselves, then $S\otimes R$ is defined as $\{s\otimes r | s\in S,r\in R\}$.
\begin{lemma}\label{lem:proj_tensor}
	Let $S\subseteq \M_d$ and $R\subseteq \M_{d'}$. If, for every $s\in S$, there is a projection $p\in \Span(S)$ such that $ps=s$, and for every $r\in R$ there is a projection $q\in \Span(R)$ such that $rq=r$, then $\alg(S)\otimes \alg(R)=\alg(S\otimes R)$. If there also exist such $p,q$ for for all $s\in S^*=\{s^*|s\in S\}$ and all $r\in R^*$, then $\alg^*(S)\otimes\alg^*(R)=\alg^*(S\otimes R)$.
\end{lemma}
\begin{proof}
	Note that for any two sets $S$ and $R$, $S\subseteq \alg(S)$ and $R\subseteq \alg(R)$, so $S\otimes R\subseteq \alg(S)\otimes \alg(R)$, so $\alg(S\otimes R)\subseteq \alg(S)\otimes \alg(R)$.\par
	For the reverse inclusion, free products of elements of $S$ will span $\alg(S)$, and free products of elements of $R$ will span $\alg(R)$. The tensor product of spanning sets is a spanning set of the tensor product, so elements of the form $s_1\cdots s_n\otimes r_1\cdots r_m$, with $s_i\in S$ and $r_j\in R$, will span $\alg(S)\otimes\alg(R)$. We take an arbitrary element of this form, $s_1\cdots s_n\otimes r_1\cdots r_m$, and then take $p\in \Span(S)$ such that $ps_1=s_1$ and $q\in \Span(R)$ such that $r_mq=r_m$. Then:
	\begin{align*}
	s_1\cdots s_n\otimes r_1\cdots r_m=&ps_1\cdots s_n\otimes r_1\cdots r_mq\\
	=&(p\otimes r_1\cdots r_m)(s_1\cdots s_n\otimes q)\\
	=&(p\otimes r_1)(p\otimes r_2)\cdots (p\otimes r_m)(s_1\otimes q)(s_2\otimes q)\cdots (s_n\otimes q)
	\end{align*}
	Since $p\in \Span(S)$, then there are elements $\{s'_j\}_{j=1}^n$ in $S$ such that $p=\sum_{j=1}^na_js'_j$. But for any $r_i$, $s'_j\otimes r_i\in S\otimes R$, so the sum $\sum a_j(s'_j\otimes r_i)=p\otimes r_i$ is in  $\alg(S\otimes R)$. Similarly, $s_i\otimes q\in \alg(S\otimes R)$. Thus, the product above is also in $\alg(S\otimes R)$, and thus all the basis elements of $\alg(S)\otimes \alg(R)$ are in $\alg(S\otimes R)$, so $\alg(S)\otimes \alg(R)\subseteq alg(S\otimes R)$.\par
	For the *-algebras, a very similar logic holds. Free products of the form $s_1\cdots s_n\otimes r_1\cdots r_m$, $s_i\in S\cup S^*$ and $r_i\in R\cup R^*$, will span $\alg^*(S)\otimes\alg^*(R)$. Take and arbitrary element of this form and let $p$ and $q$ be projections defined as before, i.e., $ps_1=s_1$ and $r_mq=r_m$. Suppose $r_i$ is in $R^*$, so $r_i=\tilde{r}_i^*$, with $\tilde{r}_i\in R$. Projections are self-adjoint, so $p\otimes r_i=p^*\otimes \tilde{r}_i^*=(p\otimes \tilde{r}_i)^*$ is in $\alg^*(S\otimes R)$. Similarly, if $s_i=\tilde{s}_i^*$ is in $S^*$, then $s_i\otimes q=(\tilde{s}_i\otimes q)^*\in \alg^*(S\otimes R)$. Thus the same decomposition can be done as the one above:
	\[s_1\cdots s_n\otimes r_1\cdots r_m=(s_1\otimes q)\cdots (s_n\otimes q)(p\otimes r_1)\cdots (p\otimes r_m)\]
	And since all of the terms on the right-hand side are in $\alg^*(S\otimes R)$, then $\alg^*(S)\otimes\alg^*(R)=\alg^*(S\otimes R)$.
\end{proof}

\begin{theorem}\label{thm:md_tensor}
For any two unital quantum channels $\UCP_1,\UCP_2$, 
\[\M_{\UCP_1\otimes\UCP_2}=\M_{\UCP_1}\otimes \M_{\UCP_2}.\]
\end{theorem}
\begin{proof}
	 Let $\UCP_1(x)=\displaystyle\sum_{i=1}^ma_ixa_i^*$ and $\UCP_2(x)=\displaystyle\sum_{i=1}^nb_ixb_i^*$ be the Kraus decomposition of $\UCP_1$ and $\UCP_2$ respectively. Trace preservation implies $1=\displaystyle\sum_{i=1}^m a_i^*a_i=\displaystyle\sum_{j=1}^{n}b_{j}^*b_{j}$. \par
	The Kraus operators of $\UCP_1^*\circ\UCP_1$ are $\{a_i^*a_j\}$ for any $i,j$. Define $S=\{a_i^*a_j|1\leq i,j\leq m\}$. Similarly, let $R=\{b_i^*b_j|1\leq i,j\leq n\}$ be the set Kraus operators of $\UCP_2^*\circ\UCP_2$. Then the Kraus operators of $\UCP_1^*\circ\UCP_1\otimes\UCP_2^*\circ\UCP_2 (=(\UCP_1\otimes\UCP_2)^*\circ(\UCP_1\otimes\UCP_2))$ are $\{a_i^*a_j\otimes b_k^*b_l:1\leq i,j\leq m \ \text{and} \ 1\leq k,l\leq n\}$, or $S\otimes R$. Since $1\in \Span(S)$ and $1\in \Span(R)$,  we have the necessary projections to use Lemma \ref{lem:proj_tensor}.  Hence we have that
	\[\alg^*(S\otimes R)=\alg^*(S)\otimes \alg^*(R).\]
Now the finite dimensional *-algebras are von Neumann algebras and by the commutant-tensor product theorem for von Neumann algebras ( \cite{kadison-ringroseII}, Theorem 11.2.16) we have that \[\alg^*(S\otimes R)'={\alg^*(S)}'\otimes {\alg^*(R)}'.\] Then by Theorem \ref{cummutant-fixed pnt} $\alg^*(S\otimes R)'=\rm{Fix}_{(\UCP_1\otimes\UCP_2)^*\circ(\UCP_1\otimes\UCP_2)}$ and $\alg^*(S)'=\rm{Fix}_{\UCP_1^*\circ\UCP_1}$ and $\alg^*(R)'=\rm{Fix}_{\UCP_2^*\circ\UCP_2}$, thus 
	\[\rm{Fix}_{(\UCP_1\otimes\UCP_2)^*\circ(\UCP_1\otimes\UCP_2)}=\rm{Fix}_{\UCP_1^*\circ\UCP_1}\otimes \rm{Fix}_{\UCP_2^*\circ\UCP_2}.\]
Now invoking Theorem \ref{kribs-spekkens} and noting that $\M_{\UCP_1\otimes\UCP_2}=\rm{Fix}_{(\UCP_1\otimes\UCP_2)^*\circ(\UCP_1\otimes\UCP_2)}$ we immediately obtain
\[\M_{\UCP_1\otimes\UCP_2}=\M_{\UCP_1}\otimes \M_{\UCP_2}.\]
\end{proof}

Since the multiplicative domain behaves well with the tensor product, it leads to a simple form for the multiplicative index:
\begin{prop}\label{prop:kappa_tensor_bound}
	Given two unital channels $\UCP_1:\M_d\rightarrow \M_d,\UCP_2:\M_{d'}\rightarrow \M_{d'}$, then $\kappa(\UCP_1\otimes\UCP_2)=\max\{\kappa(\UCP_1),\kappa(\UCP_2)\}$ (where $\kappa$ is the multiplicative index).
\end{prop}
\begin{proof}
	If $k\geq \max\{\kappa(\UCP_1),\kappa(\UCP_2)\}(=:\kappa_{\max})$, then:
	\[\M_{(\UCP_1\otimes\UCP_2)^k}=\M_{\UCP_1^k\otimes\UCP_2^k}=\M_{\UCP_1^k}\otimes \M_{\UCP_2^k}=\M_{\UCP_1^\infty}\otimes \M_{\UCP_2^\infty}.\]
	That is, the multiplicative domain is constant after $\kappa_{\max}$, so $\kappa(\UCP_1\otimes\UCP_2)\leq\kappa_{\max}$. Then suppose $k<\kappa_{\max}$ (and, without loss of generality, suppose $\kappa(\UCP_2)=\kappa_{\max}$). By a similar logic:
	\[\M_{(\UCP_1\otimes\UCP_2)^k}=\M_{\UCP_1^k}\otimes \M_{\UCP_2^k}\subsetneq \M_{\UCP_1^{k+1}}\otimes \M_{\UCP_2^{k+1}}=\M_{(\UCP_1\otimes\UCP_2)^{k+1}}.\]
	Since the multiplicative domain is still strictly decreasing with $k$, then $\kappa(\UCP_1\otimes\UCP_2)>k$ and the result follows.
\end{proof}
The above proposition implies the following corollary:
\begin{corollary}\label{splitting stabilising algebra}
For unital channels $\UCP_1,\UCP_2$ we have 
\[\M_{{(\UCP_1\otimes\UCP_2)}^{\infty}}=\M_{\UCP_1^\infty}\otimes\M_{\UCP_2^\infty}.\]
\end{corollary}
\subsection{Fixed Points of Product Channels}\label{sec:fixed_points_of_products}
For a unital channel $\UCP_1\otimes\UCP_2$, the fixed point set $\rm{Fix}_{\UCP_1\otimes\UCP_2}$ is a subalgebra of $\M_d\otimes\mathcal{M}_{d'}$ and unlike the multiplicative domain case, this subalgebra does not split nicely. However, using Theorem \ref{thm:md_tensor}, we can provide an exact description
of this algebra and characterize when this subalgebra splits and recapture the result of \cite{Luczak}. Our results are specific cases of \cite{Watanabe} and \cite{Luczak}, but through a vastly different approach. The spectrum of the tensor product of two channels is known to be the set product of the two spectra, but this theorem characterizes the eigen operators as only the obvious choices. In what follows $\mathbb{T}$ represents the unit circle in the complex plane. Note that (see \cite{wolf}) for any quantum channel $\UCP$, all the eigenvalues lie in the closed unit disc of the complex plane. We define the spectrum ($\rm{Spec}_{\UCP}$) of $\UCP$ as follows
\[\rm{Spec}_{\UCP}=\{\lambda\in \mathbb{C} \ | \ (\lambda 1-\UCP) \ \text{is \ not \ \  invertible \ on} \ \M_d \},\] 
where $1$ is the identity operator on $\Md$. The set $\rm{Spec}_{\UCP}\cap \mathbb{T}$ is called the \emph{peripheral eigenvalues} and the corresponding eigenoperators are called \emph{peripheral eigenvectors}.
\begin{theorem}\label{thm:fixed_point_products}
Let $\UCP_1:\M_d\rightarrow \M_d,\UCP_2:\M_{d'}\rightarrow \M_{d'}$ be two unital quantum channels. Then for any $\lambda\in \mathbb{T}$:
\[\leftset z \in \M_d\otimes \M_{d'} \midsetr \UCP_1\otimes\UCP_2(z)=\lambda z\rightset = \Span\leftset x_i\otimes y_i \midsetr \UCP_1(x_i)=\mu_1 x_i,\UCP_2(y_i)=\frac{\lambda}{\mu_1} y_i\rightset.\]
\end{theorem}
\begin{proof}
Let $\lambda\in \mathbb{T}$ . For the left inclusion, suppose there are two numbers $\mu_1,\mu_2$ such that $\UCP_1(x)=\mu_1 x$ and $\UCP_2(y)=\mu_2 y$ for matrices $x,y$ and $\lambda=\mu_1\mu_2$. Then
\[\UCP_1\otimes\UCP_2(x\otimes y)=(\mu_1 x)\otimes (\mu_2 y)=\lambda (x\otimes y).\]
For the right inclusion, let $z$ be a matrix such that $\UCP_1\otimes \UCP_2(z)=\lambda z$. By Theorem 2.5 from \cite{miza}, we know that the peripheral eigenvectors of a channel are precisely the stabilized multiplicative domain. Thus:
\[z\in \M_{(\UCP_1\otimes\UCP_2)^\infty}=\M_{\UCP_1^\infty}\otimes \M_{\UCP_2^\infty}.\]
We can then represent $z$ as $z=\sum_{i=1}^mx_i\otimes y_i$, where $x_i\in \M_{\UCP_1^\infty}$ and $y_i\in \M_{\UCP_2^\infty}$. By the same theorem, we know that $x_i$ is a linear combination of peripheral eigenvectors of $\UCP_1$. Thus we can further decompose $z$ as
\[z=\sum_{i=1}^{m'}x_i'\otimes y_i\]
where the $\{x_i'\}$ are linearly independent and $\UCP_1(x_i')=\mu_ix_i'$ with $\mu_i\in\mathbb{T}$. This gives us:
\[\UCP_1\otimes\UCP_2(z)=\sum_{i=1}^{m'}\mu_i x_i'\otimes \UCP_2(y_i).\]
But by choice of $z$, $\UCP_1\otimes \UCP_2(z)=\lambda z = \lambda \sum_{i=1}^{m'}x_i'\otimes y_i$. By the linear independence of $\{x_i'\}$, we have that $\lambda y_i=\mu_i\UCP_2(y_i)$, i.e., $\UCP_2(y_i)=\frac{\lambda}{\mu_i}y_i$. This holds for all $i$, giving the required inclusion.
\end{proof}
Using the above theorem we obtain the following corollary which first appeared in \cite{Luczak}, Corollary 13 in a more general context. However our method of obtaining this result is significantly different from \cite{Luczak}.
\begin{corollary}\label{cor:fixed_point_split}
For two unital channels $\UCP_1$ and $\UCP_2$ with spectra $\rm{Spec}_{\UCP_1}$ and $\rm{Spec}_{\UCP_2}$ respectively, the fixed point algebra splits if and only if the intersection of the peripheral spectra is trivial. That is, 
\[\rm{Fix}_{\UCP_1\otimes\UCP_2}=\rm{Fix}_{\UCP_1}\otimes
\rm{Fix}_{\UCP_2},\]
if and only if $\rm{Spec}_{\UCP_1}\cap \rm{Spec}_{\UCP_2}\cap\mathbb{T}
=\{1\}$.
\end{corollary}
\begin{proof}
The fixed points are the special case of peripheral eigen-operators where $\lambda=1$. Using Theorem \ref{thm:fixed_point_products}, we have that the fixed points are given by
\[\rm{Fix}_{\UCP_1\otimes\UCP_2}=\text{span}\leftset x_i\otimes y_i\midsetr \UCP_1(x_i)=\mu x_i,\UCP_2(y_i)=\overline{\mu}y_i,\vert\mu\vert=1\rightset.\]
This set will equal $\rm{Fix}_{\UCP_1}\otimes\rm{Fix}_{\UCP_2}$ if and only if there is no $\mu\in\rm{Spec}_{\UCP_1}\cap\mathbb{T}$ with $\overline{\mu}\in\rm{Spec}_{\UCP_2}\cap\mathbb{T}$. Since the spectrum of a quantum channel is closed under conjugation, this means $\mu$ would need to be in both spectra. Thus, the spectrum will split if and only if the intersection of the spectra is trivial.
\end{proof}
Theorem \ref{thm:fixed_point_products} is particularly helpful to analyze the ergodicity or irreducibility of tensor product of quantum channels. We provide the definition of such channels below:
\begin{definition}
A channel $\UCP:\M_d\rightarrow \M_d$ is called irreducible if there is no non-trivial projection $p\in \M_d$ such that $\UCP(p)\leq \lambda p$, for $\lambda>0$.
\end{definition}
\begin{definition}
An irreducible channel $\UCP$ is called primitive if the set of peripheral eigenvalues contains only 1, that is if $\rm{Spec}_\UCP\cap \mathbb{T}=\{1\}$.
\end{definition}
We note down some properties of irreducible positive linear maps:
\begin{theorem}{\rm{(see \cite{evans-krohn})}}\label{thm:evans-krohn}
Let $\UCP$ be a positive linear map on $\Md$ and let $r$ be its spectral radius. Then
\begin{enumerate}
\item There is a non zero positive element $x \in \M_d$ such that $\UCP(x)=rx$
\item If $\UCP$ is irreducible and if a positive $y \in {\Md}$ is an eigenvector of $\UCP$ corresponding to some
eigenvalue $s$ of $\UCP$, then $s = r$ and $y$ is a positive scalar multiple of $x$.
\item If $\UCP$ is unital, irreducible and satisfies the Schwarz inequality for positive linear maps then
\begin{itemize}
\item $r=1$ and $\rm{Fix}_{\UCP}=\mathbb{C}1$.
\item Every peripheral eigenvalue $\lambda \in \rm{Spec_{\UCP}\cap \mathbb{T}} $ is simple and the corresponding eigenspace is spanned by a unitary $u_{\lambda}$ which satisfies $\UCP(u_{\lambda}x)=\lambda u_{\lambda}\UCP(x)$, for all $x\in \Md$.
\item The set $\Gamma=\rm{Spec_{\UCP}\cap\mathbb{T}}$ is a cyclic subgroup of the group $\mathbb{T}$ and the corresponding eigenvectors form a cyclic group which is isomorphic to $\Gamma$ under the isomorphism $\lambda \rightarrow u_{\lambda}$. 
\end{itemize}
\end{enumerate}
\end{theorem}
Often irreducible channels are called \emph{ergodic channels}. Ergodic/irreducible positive maps have been a great topic of interest (see \cite{evans-krohn}, \cite{farenick}, \cite{ergodic}). The study of such maps enriched the analysis of non-commutative Perron-Frobenius theory. Although ergodicity of a quantum dynamical system (discrete or continuous) has received much attention, the same analysis in the tensor product framework has been talked about less except \cite{Watanabe} and \cite{Luczak}. Here we present necessary and sufficient conditions for a channel to be irreducible and primitive in the tensor product system. By the aid of Theorem \ref{thm:fixed_point_products} we recapture Theorem 5.3 in \cite{Watanabe}.
\begin{theorem}
Let $\UCP_1$ be an irreducible unital quantum channel with $n$ peripheral eigenvalues $\Gamma_n$. Then:
\begin{enumerate}
\item
The product $\UCP_1\otimes\UCP_1$ is irreducible if and only if $\UCP_1$ is also primitive, in which case $\UCP_1\otimes\UCP_1$ is also primitive.
\item
For any primitive unital channel $\UCP_2$, $\UCP_1\otimes\UCP_2$ is irreducible.
\item
If $\UCP_2$ is irreducible with $m$ peripheral eigenvalues $\Gamma_m$, then $\UCP_1\otimes \UCP_2$ is irreducible if and only if $gcd(n,m)=1$.
\end{enumerate}
\end{theorem}
\begin{proof}
(1) For $\UCP_1\otimes\UCP_1$ to be irreducible, its fixed points would need to be $\C (1\otimes 1)$, meaning the fixed points would have to split. By Corollary \ref{cor:fixed_point_split}, this would occur if and and only if the peripheral spectrum of $\UCP_1$ is trivial, meaning $\UCP_1$ is primitive. Since the spectrum of a quantum channel is contained in the unit disc, in this case the peripheral spectrum of $\UCP_1\otimes\UCP_1$ will still be trivial and thus it will be primitive. \par
(2) Since $\UCP_2$ is primitive, its only eigenvalue is 1 with eigenvector $1$. Thus the fixed points of $\UCP_1\otimes\UCP_2$ will split, and since both fixed point algebras are trivial, the product will also be trivial.

For item (3), if $gcd(n,m)=1$, then the two cyclic groups $\Gamma_n, \Gamma_m$ intersect trivially and hence by Corollary \ref{cor:fixed_point_split} we get  $\rm{Fix}_{\UCP_1\otimes\UCP_2}=\rm{Fix}_{\UCP_1}\otimes\rm{Fix}
_{\UCP_2}=\mathbb{C}1\otimes\mathbb{C}1=\mathbb{C}(1\otimes 1)$. 
\newline
Conversely, if $\UCP_1\otimes\UCP_2$ is irreducible, then $\rm{Fix}_{\UCP_1\otimes\UCP_2}=\mathbb{C}1$. From Theorem \ref{thm:evans-krohn} we know that the peripheral spectrum of $\UCP_1\otimes\UCP_2$ is a cyclic subgroup of some order $N$.
Since $\rm{Fix}_{\UCP_1}=\mathbb{C}1=\rm{Fix_{\UCP_2}}$, it is evident that this can only happen 
if the fixed point algebra splits. By Corollary \ref{cor:fixed_point_split} again we conclude that $\Gamma_n\cap\Gamma_m=\{1\}$; that is, $gcd(n,m)=1$.  
\end{proof}
Theorem \ref{thm:fixed_point_products} gives structure to the eigenspaces of these eigenvalues. For some intuition on this, a channel acts like an automorphism on its stabilized multiplicative domain, so in some sense it is ``normal'' on this subalgebra. The eigenspaces of the tensor product of two normal matrices will simply be the products of the original eigenspaces, and here something similar holds for the ``normal part'' of the channel.    


\section{Restrictions on the Multiplicative Index}\label{sec:restrictions_on_kappa}
\subsection{Maximal Unital Proper *-Subalgebras (MUPSAs)}\label{sec:MUPSAS}
Proposition \ref{prop:kappa_tensor_bound} restricts which channels can be product channels, since the multiplicative index must be the same as the multiplicative index of one of the channels in the product. Our goal is thus to restrict the possible values of the multiplicative index. An obvious bound is the dimension of the matrix algebra, $d^2$, but in fact we can do much better by looking at chains of maximal unital proper *-subalgebras, defined in the obvious way as follows:
\begin{definition}
	An algebra $\A$ is a maximal unital proper *-subalgebra (for convenience, a ``MUPSA") of a C*-algebra $\B$ if $\A$ is unital proper *-subalgebra of $\B$ (meaning $\A\neq \B$, $1\in\A$, and $\A^*=\A$) such that if $\tilde{\A}$ is another unital proper *-algebra with $\A\subseteq \tilde{\A}\subseteq\B$, then either $\tilde{\A}=\B$ or $\tilde{\A}=\A$.
\end{definition}
While there are many possible forms of a subalgebra of $\M_d$, restricting to MUPSAs allows us to precisely characterize their structure, up to isomorphism.
We use the Wedderburn decomposition extensively. For a matrix algebra $\A$, one can always decompose it as
\[\A\cong \bigoplus_{r=1}^m \M_{n_r}\otimes 1_{k_r}.\]
This is the Wedderburn decomposition.
\begin{lemma}\label{lem:MdMUPSA}
	If $\A$ is a MUPSA of $\M_d$, then (up to isomorphism) $\A=\M_{d-r}\oplus \M_r$, where $1\leq r\leq d-1$.
\end{lemma}
\begin{proof}
	Let $\A$ be a *-subalgebra of $\M_d$. Then let 
	\[\A=\bigoplus_{r=1}^m\M_{n_r}\otimes 1_{k_r}\]
	be the Wedderburn decomposition of $\A$. If $m\geq 3$, then the following subalgebra
	\[\tilde{\A}=\M_{n_1}\otimes 1_{k_1}\oplus \M_{\sum_{r=2}^mn_rk_r}\]
	will strictly contain $\A$, but be strictly contained in $\M_d$, contradicting the maximality of $\A$. If $m=1$, then $\A=\M_{d/p}\otimes 1_p$ (for some number $p$ dividing $d$). Then $\A\subsetneq \M_{d/p}\oplus \M_{d(p-1)/p}$, contradicting maximality of $\A$. Thus $m=2$, and $\A=\M_{n_1}\otimes 1_{k_1}\oplus \M_{n_2}\otimes 1_{k_2}$. If $k_1>1$, then $\A$ is a proper subalgebra of 
	\[\underbrace{\M_{n_1}\oplus\cdots\oplus \M_{n_1}}_\text{$p$ times} \oplus \M_{n_2}\otimes 1_{k_2}\]
	which in turn is a proper subalgebra of $\M_d$, again contradicting maximality. The same argument applies to $k_2$, and thus
	\[\A=\M_{n_1}\oplus \M_{n_2}.\]
	Since $\A$ is unital, $n_1+n_2=d$, so we can write $n_1=d-r$ and $n_2=r$, for $0\leq r\leq d$. If $r=0$ or $r=d$, then $\A=\M_d$, so $1\leq r\leq d-1$.
\end{proof}
\begin{theorem}\label{thm:any_MUPSA}
	Let $\B$ be a unital *-subalgebra of $\M_d$ with Wedderburn decomposition $\B=\bigoplus_{r=1}^m\M_{n_r}\otimes 1_{k_r}$. If $\A$ is a MUPSA of $\B$, then, up to unitary equivalence, $\A$ has one of the following forms:
	\begin{enumerate}
		\item
		\[\A=\left(\M_{n_j-s}\otimes 1_{k_j}\right)\oplus\left(\M_{s}\otimes 1_{k_j}\right)\oplus\bigoplus_{r=1,r\neq j}^m\M_{n_r}\otimes 1_{k_r}\]
		for some $1\leq j\leq m$ and some $s$ such that $1\leq s\leq n_j-1$, or
		\item
		\[\A=\left(\M_{n_j}\otimes 1_{k_j+k_i}\right)\oplus\bigoplus_{r=1,r\neq i,j}^m\M_{n_r}\otimes 1_{k_r}\]
		for some $i,j$ such that $1\leq i,j\leq m$ and $n_j=n_i$.
	\end{enumerate}
\end{theorem}
Before the proof, we recall a result of Bratteli's from \cite{bratteli2} that will be very useful. 
\begin{prop}\label{prop:bratteli}
Let $\A\cong\oplus_{r=1}^\ell \M_{a_r}$ and $\B\cong\oplus_{r=1}^{m}\M_{n_r}$ as algebraic isomorphisms, with $\A\subseteq\B$. Then there exist integers $p_{rs}\in\N\cup\{0\}$ for $r\in\{1,\cdots,m\}$ and $s\in\{1,\cdots,\ell\}$ such that we can identify $\A$ with
\[\bigoplus_{r=1}^{m}\left(\bigoplus_{s=1}^\ell \M_{a_s}\otimes 1_{p_{rs}}\right),\]
with the convention that, for any two matrix algebras $\M_{n_1}$ and $\M_{n_2}$, $\M_{n_1}\oplus (\M_{n_2}\otimes 1_0)=\M_{n_1}$.
\end{prop}
This is an informal statement of the proposition, but it says that every block $\M_{a_r}$ in $\A$ is embedded into zero or more blocks of $\B$. Note that the equivalences ignore the tensor factors in the usual Wedderburn decomposition, since these affect only the norms, not the algebraic structure. Hence, to prove Theorem \ref{thm:any_MUPSA} we will first use Bratteli's result for the algebraic structure, then recover the norms.
\begin{lemma}\label{lem:bratteli_MUPSA}
Let $\A$ and $\B$ be matrix algebras such that $\A$ is a MUPSA of $\B$, with
\[\B=\bigoplus_{r=1}^m\M_{n_r},\, \A\cong\bigoplus_{r=1}^\ell\M_{a_r},\]
and the embedding of ${\A}$ into ${\B}$ has the form
\[\bigoplus_{r=1}^{m}\left(\bigoplus_{s=1}^\ell \M_{a_s}\otimes 1_{p_{rs}}\right).\]
Then, up to a permutation of the blocks of $\A$, either:
\begin{enumerate}
\item
The number of blocks in $\A$ is $m+1$, and there is an index $j\in[m]$, such that: 
\begin{itemize}
\item
For all $r\neq j$, $\M_{a_r}=\M_{n_r}$ and $p_{rs}=\delta_{rs}$. 
\item
There is some $t$ with $1\leq t\leq n_j-1$ such that $\M_{a_j}=\M_t$, $\M_{a_{m+1}}=\M_{n_j-t}$, $p_{js}=\delta_{js}+\delta_{(m+1)s}$.
\end{itemize}
\item
The number of blocks in $\A$ is $m-1$, and there are indices $i,j\in[m]$, $j<i$, such that:
\begin{itemize}
\item
For all $r< i$, $\M_{a_r}=\M_{n_r}$ and $p_{rs}=\delta_{rs}$. 
\item
For all $r> i$, $\M_{a_{r-1}}=\M_{n_{r}}$ and $p_{rs}=\delta_{(r-1)s}$.
\item
$\M_{n_i}=\M_{n_j}$ and $p_{is}=\delta_{sj}$. 
\end{itemize}
\end{enumerate}
\end{lemma}
This lemma states that, with one or two exceptions, every block of $\A$ maps surjectively into a block of $\B$. For the remaining block(s), either there are two blocks of $\A$ that map into one block of $\B$, or there is one block of $\A$ that maps to two blocks of $\B$.\par
Note that we assume $\B$ is equal to the structure without tensor products, but we can only assume $\A$ is isomorphic to  such a structure. The decomposition of $\A$ given in the statement of Lemma \ref{lem:bratteli_MUPSA} ignores the dimension, and the embedding into $\B$ may not be isometric. Indeed, if case 2 holds, then one block of $\A$ will contain a tensor product with $1_2$.
\begin{proof}
For all $r\in[m]$, define $\A_{n_r}$ as the $r$th block of the embedding of $\A$, i.e.:
\begin{equation}
\A_{n_r}=\bigoplus_{s=1}^\ell \M_{i_s}\otimes 1_{p_{rs}}\subseteq \M_{n_r}.\label{eq:mupsa1}
\end{equation}
With this notation, we have that
\[\A_\B\subseteq\bigoplus_{r=1}^m \A_{n_r}\subseteq {\B}\]
where ${\A}_{\B}$ is the image of ${\A}$ of the embedding into ${\B}$. Note that $\A_{\B}$ must also be a MUPSA of $\B$.\par
For each $r$, $\A_{n_r}$ may be a proper subalgebra of $\M_{n_r}$ or not. Suppose there is some $j$ where it is a proper subalgbera.  Then we can take the subalgebra $\tilde{\A}$ defined by
\[\tilde{\A}=\A_{n_{j}}\oplus\bigoplus_{r\neq j}\M_{n_r}\]
and this will be a proper subalgebra of ${\B}$ and it will contain ${\A}_{\B}$. Since $\A_{\B}$ is also a MUPSA, $\tilde{\A}=\A_{\B}$. Thus, $\A$ must have the form of $\tilde{\A}$, so $\A$ can have at most one $j$ such that $\A_{n_{j}}$ is a proper subalgebra of $\M_{n_{j}}$.\par
In this case, we can argue that $\A_{n_{j}}$ must itself be a MUPSA of $\M_{n_{j}}$, or ${\A}$ would not be maximal - we could take a MUPSA as the $j$th block instead. By Lemma \ref{lem:MdMUPSA}, $\A_{n_{j}}$ must have the form $\M_{n_j-t}\oplus \M_t$ for some $t$ with $1\leq t\leq n_j-1$. This proves Part (1).\par
The other possible situation is where $\A_{n_r}=\M_{n_r}$ for all $r$. This means that in the notation of Equation \ref{eq:mupsa1}, there can only be one block of $\A$ in each block of $\B$, so for each $r$, there is a unique $s_0(r)$ such that $p_{rs_0(r)}=1$, and $p_{rs}=0$ for all $s\neq s_0(r)$. This means that the embedding of $\A$ into $\B$ looks like
\[\bigoplus_{s=1}^\ell \M_{a_s}\mapsto\bigoplus_{r=1}^{m}\left(\M_{a_{s_0(r)}}\right).\]
The direct sum on the left is not all of $\A$, it is only isomorphic to $\A$. A block on the left might appear twice in the embedding if there is some $i\neq j$ such that $s_0(i)=s_0(j)$. This is how, even though each block is surjectively covered by the embedding, $\A$ can still be a proper subalgebra of $\B$, since $\B$ has more freedom between blocks.\par
If $\ell=m$, then each block of $\A$ embeds surjectively into each block of $\B$, implying the contradictory statement that $\A=\B$. Thus $\ell<m$. This means there must be some $i$ and $j$ such that $s_0(i)=s_0(j)$. That is, some block of $\A$ maps to two blocks in $\B$. We define an algebra $\tilde{\A}$ with $\A\subseteq\tilde{\A}\subseteq\B$ such that 
\[\tilde{\A}=\bigoplus_{s\neq i,s\leq m}\M_{n_s}\mapsto \bigoplus_{r=1}^m\left((\M_{n_r}\otimes 1_{1-\delta_{ri}})\oplus (\M_{n_j}\otimes 1_{\delta_{ri}})\right).\]
That is, $\tilde{\A}$ is just all of the blocks of $\B$ except the $i$th block; to embed it into $\B$, we use the identity on all blocks, and send a copy of the $j$th block to the $i$th block of $\B$. Since we required that each block $\A_{n_j}=\M_{n_j}$, then $\M_{n_j}=\M_{n_i}$. Clearly, $\tilde{\A}$ is a proper subalgebra of $\B$, and by this construction, $\tilde{\A}$ must contain $\A$. Hence $\tilde{\A}=\A$.\par
Thus a MUPSA must have the form of $\tilde{\A}$ for some blocks $i$ and $j$, hence $\ell=m-1$ and in all other blocks, $\A$ and $\B$ are equal. This proves Part (2).
\end{proof}
\begin{proof}[Proof of Theorem \ref{thm:any_MUPSA}]
Given 
\[\B=\bigoplus_{r=1}^m\M_{n_r}\otimes 1_{k_r}\]
we can define a new algebra $\tilde{\B}$ as
\[\tilde{\B}=\bigoplus_{r=1}^m\M_{n_r}.\]
This will be *-isomorphic, but not isometric, to $\B$. The natural isomorphism $\phi:\tilde{\B}\rightarrow \B$ can be defined as
\[\phi(x_1,\cdots,x_m)=(x_1\otimes 1_{k_1},\cdots,x_m\otimes 1_{k_m}).\]
Then we can let $\tilde{\A}=\phi^{-1}(\A)\subseteq\tilde{\B}$. In fact, $\tilde{\A}$ will be a MUPSA of $\tilde{\B}$, since any subalgebra of $\tilde{\B}$ can map to a subalgebra of $\B$.\par
Then, ignoring tensor products, we can write
\[\tilde{\A}\cong\bigoplus_{r=1}^\ell\M_{a_r}\]
and apply Lemma \ref{lem:bratteli_MUPSA} and consider the two cases. \par
In the first case, $\ell=m+1$, and the decomposition of $\tilde{\A}$ looks like
\[\tilde{\A}\mapsto \M_{t}\oplus\M_{n_j-t}\oplus\bigoplus_{r\neq j;r\leq m}\M_{n_r}.\]
By dimension counting, this must actually equal $\tilde{\A}$, so
\[\tilde{\A}= \M_{t}\oplus\M_{n_j-t}\oplus\bigoplus_{r\neq j;r\leq m}\M_{n_r}.\]
Then we can write
\[\A=\phi(\tilde{\A})=(\M_t\otimes 1_{k_j})\oplus(\M_{n_j-t}\otimes 1_{k_j})\oplus\bigoplus_{r\neq j;r\leq m}\M_{n_r}\otimes 1_{k_r},\]
thus proving Part (1).\par
In the second case, $\ell=m-1$ and the embedding of $\tilde{\A}$ into $\tilde{\B}$ is
\[\tilde{\A}\mapsto (\M_{n_j}\otimes 1_2)\oplus\bigoplus_{r\neq i,j;r\leq m}\M_{n_r}.\]
Here we've used the fact that $\M_{n_j}$ maps into two blocks and replaced these two blocks with a tensor product with $1_2$. Once again, by dimension counting, this is not just an embedding, it is the actual structure of $\tilde{\A}$. Hence, we can write $\A$ as $\phi(\tilde{\A})$. To handle the $j$th block, note that any element of $\tilde{\A}$ has the same elements in the $i$ and $j$ components, so it looks like $(x_j,x_j)$. When we apply $\phi$ to these components, they become $(x_j\otimes 1_{k_i},x_j\otimes 1_{k_j})=x_j\otimes 1_{k_i+k_j}$. Thus,
\[\A=\phi(\tilde{\A})=(\M_{n_j}\otimes 1_{k_i+k_j})\oplus\bigoplus_{r\neq i,j;r\leq m}\M_{n_r}\otimes 1_{k_r}.\]
This proves part (2).
\end{proof}

This characterization of MUPSAs also characterizes the lattice of proper *-algebras of $\M_d$. For example, if $d=4$, then the MUPSAs form the lattice shown in Figure \ref{fig:lattice_diagram}. Note that in the figure, the length of the longest chain of subalgebras, including $\M_4$, is 7. The next lemma generalizes this.
\begin{figure}[H]
	\caption{The lattice of unital *-subalgebras of $\M_4$. Since $\M_1=\C$, the minimal element is $\C 1_4$.}
\begin{center}
\begin{tikzpicture}
	\node (M4) at (2,7) {$\M_4$};
	\node (M3M1) at (0,6) {$\M_3\oplus \M_1$};
	\node (M2M2) at (4,6) {$\M_2\oplus \M_2$};
	\node (M2M1M1) at (0,5) {$\M_2\oplus \M_1\oplus \M_1$};`
	\node (M1M1M1M1) at (0,4) {$\M_1\oplus \M_1\oplus \M_1\oplus \M_1$};
	\node (M1x2M1M1) at (0,3) {$\M_1\otimes 1_2\oplus \M_1\oplus \M_1$};
	\node (M2x2) at (4,4) {$\M_2\otimes 1_2$};
	\node (M1x3M1) at (0,2) {$\M_1\otimes 1_3\oplus \M_1$};
	\node (M1x2M1x2) at (4,2) {$\M_1\otimes 1_2\oplus \M_1\otimes 1_2$};
	\node (M1x4) at (2,1) {$\M_1\otimes 1_4$};
	\draw [->] (M4) edge (M3M1) edge (M2M2)
	(M3M1) edge (M2M1M1)
	(M2M2) edge (M2x2) edge (M2M1M1)
	(M2x2) edge (M1x2M1x2)
	(M1x2M1x2) edge (M1x4)
	(M2M1M1) edge (M1M1M1M1)
	(M1M1M1M1) edge (M1x2M1M1)
	(M1x2M1M1) edge (M1x3M1) edge (M1x2M1x2)
	(M1x3M1) edge (M1x4);
\end{tikzpicture}
\end{center}
\label{fig:lattice_diagram}
\end{figure}
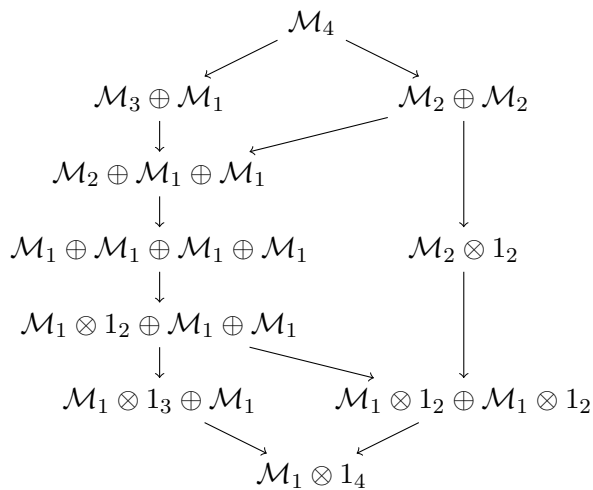
\begin{lemma}\label{lem:chain_length}
	Let $\{\A_1,\cdots,\A_n=\C 1\}$ be a descending chain of unital subalgebras of $\M_d$ and let $\A_1=\bigoplus_{r=1}^m\M_{n_r}\otimes 1_{k_r}$. Then the length of the chain is at most $\sum_{r=1}^m(2n_r-1)$.
\end{lemma}
\begin{proof}
	If $\A_{i+1}$ is not a MUPSA of $\A_i$, then there must be a chain of MUPSAs going from $\A_i$ to $\A_{i+1}$, and this will only increase the length of the chain. So, without loss of generality, assume that each algebra is a MUPSA of the previous one.\par
	For $\A_i=\bigoplus_{r=1}^{m_i}\M_{n_{i_r}}\otimes 1_{k_{i_r}}$, define $\chi(\A_i)=\sum_{r=1}^{m_i}(2n_{i_r}-1)$. We will use induction on $\chi(\A_1)$ to show that the length of the chain is at most $\chi(\A_1)$.
	Since $2n_r-1\geq 1$ for all $n_r$, if $\chi(\A_1)=1$, then $\A_1=\M_1\otimes 1_d=\C 1$. Then the length of the chain is just 1, which equals $\chi(\A_1)$.\par
	Suppose the hypothesis holds for all chains starting with algebras $\A_1$ such that $\chi(\A_1)<y$ for some $y$. Suppose we have a chain with $\A_1=\bigoplus_{r=1}^m\M_{n_r}\otimes 1_{k_r}$ such that $\chi(\A_1)=y$. Then the next algebra in the chain, $\A_2$, is a MUPSA of $\A_1$, and by Lemma \ref{thm:any_MUPSA}, it has two possible forms:
	\begin{enumerate}
		\item
		$\A_2=(\M_{n_j-s}\otimes 1_{k_j})\oplus ( \M_s\otimes 1_{k_j})\oplus\bigoplus_{r=1,r\neq j}^m \M_{n_r}\otimes 1_{k_r}$. In this case:
		\begin{align*}
		\chi(\A_2)=&\sum_{r=1,\neq j}(2n_r-1)+(2(n_j-s)-1)+(2s-1)\\
		=&\sum_{r=1,\neq j}(2n_r-1)+(2n_j-1)-1\\
		=&\sum_{r=1}^m(2n_r-1)-1\\
		=&\chi(\A_1)-1,
		\end{align*}
		which is less than $y$, so we can apply induction and declare that the length of the chain $\{\A_2,\cdots, \A_n\}$ is at most $\chi(\A_2)$; adding 1 when we add $\A_1$ takes the maximum length to $\chi(\A_2)+1=\chi(\A_1)$.
		\item
		$\A_2=(\M_{n_j}\otimes 1_{k_j+k_i})\oplus\bigoplus_{r=1,r\neq j,i}^m \M_{n_r}\otimes 1_{k_r}$. Then, noting that $2n_i-1\geq 1$, that $\chi(\A_2)$ is 
		\begin{align*}
		\sum_{r=1,r\neq i,j}^m (2n_r-1) + (2n_j-1) \leq& \sum_{r=1,r\neq i,j}^m (2n_r-1)+(2n_j-1)+(2n_i-1)-1\\
		=&\sum_{r=1}^m(2n_r-1)-1\\
		=&\chi(\A_1)-1
		\end{align*}
		Again, we apply induction and add the remaining algebra to show that the length of the chain is at most $\chi(\A_1)$.
	\end{enumerate}
\end{proof}
Since the multiplicative domains of powers of a unital channel give a chain of unital *-subalgebras, Lemma \ref{lem:chain_length} gives a bound on the multiplicative index:
\begin{theorem}\label{thm:kappa_bound}
	Let $\UCP:\M_d\rightarrow \M_d$ be a unital quantum channel for $d\geq 2$. Then the multiplicative index $\kappa(\UCP)\leq 2d-2$.
\end{theorem}
\begin{proof}
	For a channel $\UCP$, 
	$\{\M_\UCP,\M_{\UCP^2},\cdots,\M_{\UCP^\kappa}\}$ is a descending chain of subalgebras, so $\kappa$ must be less than the maximum length of such a chain. As a chain of subalgebras this has a maximum length of $2d-1$ by Lemma \ref{lem:chain_length}, which could only be achieved if $\M_\UCP=\Md$. However, if $\M_\UCP=\Md$, then the channel must be unitary and $\kappa(\UCP)=1\leq 2(d-1)$. If $\M_\UCP\neq \Md$, then the length of the chain does not achieve the maximum and must be at most $2(d-1)$.
\end{proof}
The proof of Theorem \ref{thm:kappa_bound} uses only Lemma \ref{lem:chain_length} and the fact that $\M_\UCP$ is *-closed and unital but does not use the structure of the channel itself. It's possible that the structure of $\M_\UCP$ gives a tighter bound for the multiplicative index than $2(d-1)$. So far the the largest multiplicative index we have constructed is $d$. Sections \ref{sec:ETB_channels} and \ref{sec:Schur_channels} illustrate our examples. Note that in order to get a non-trivial value for $\kappa$, one must choose $\M_\UCP$ to be a proper subalgebra of $\M_d$ and not $\M_d$ itself.
So the value obtained for the length of the longest chain in Figure \ref{fig:lattice_diagram} is one more than the possible maximum value of $\kappa$ determined by the Theorem \ref{thm:kappa_bound}. 


\subsection{Example: Entanglement Breaking Channels}\label{sec:ETB_channels}
There are many equivalent definitions of an entanglement breaking channel (see \cite{entng-brkng}, \cite{stormer2008}), but for our purposes it is most convenient to define it in terms of Kraus operators. A channel $\UCP:\M_d\rightarrow \M_d$ is an entanglement breaking channel if it can be written as:
\[\UCP(x)=\sum_{i=1}^n\varphi_i\zeta_i^* x\zeta_i\varphi_i^*,\]
where $\zeta_i,\varphi_i\in\C^d$.
We note down a useful result below which can be deduced from $\cite{stormer2008}$:
\begin{theorem}{\rm{(St{\o}rmer}, \cite{stormer2008})}
Let $\UCP:\M_d\rightarrow\M_d$ be a unital entanglement breaking channel. Then the multiplicative domain of $\UCP$ is an abelian C$^*$-algebra.
\end{theorem}
Using the above theorem we get an upper bound for the multiplicative index of unital entanglement breaking channels:
\begin{prop}\label{prop:etb_bound}
If $\UCP$ is a unital entanglement breaking channel, then $\kappa(\UCP)\leq d$.
\end{prop}
\begin{proof}
Consider the chain $\{\M_\UCP,\M_{\UCP^2},\cdots,\M_{\UCP^\kappa}\}$. Since $\M_\UCP$ is abelian, it is a subalgebra of a maximal abelian subalgebra, which has the form $\M_1\oplus\cdots\oplus \M_1$ for $\M_d$. Following Lemma \ref{lem:chain_length}, the length of the chain is at most $\sum_{i=1}^d (2(1)-1)=d$, and thus that is the maximum possible value of $\kappa(\UCP)$.
\end{proof}

The vectors $\{\zeta_i\}$ that form the Kraus operators of an entanglement breaking channel need not be linearly independent, but if the channel is unital and the vectors are linearly independent, then they are orthonormal. These are the only channels we consider, since in this case the multiplicative domain is easily calculated.
\begin{prop}\label{prop:etb_md_struct}
	Let $\UCP:\M_d\rightarrow \M_d$ be a unital quantum channel such that $\UCP(x)=\sum_{i=1}^d\varphi_i\zeta_i^*x\zeta_i\varphi_i^*$, and let $\{\varphi_i\}_{i=1}^d$ and $\{\zeta_i\}_{i=1}^d$ be orthonormal. Then $\M_\UCP=\Span\{\zeta_i\zeta_i^*|1\leq i\leq d\}$.
\end{prop}

\begin{proof}
	We will use Theorem \ref{kribs-spekkens}, which states that $\M_\UCP=\rm{Fix}_{\UCP^*\circ\UCP}$. First we compute the Kraus operators of $\UCP^*\circ\UCP$, which are all the products of the form $\zeta_i\varphi_i^*\varphi_j\zeta_j^*=\delta_{ij}\zeta_i\zeta_j^*$. Hence the set of Kraus operators is $\{\zeta_i\zeta_i^*\}_{i=1}^d$. This spans a maximal abelian subalgebra, and hence its span is its own commutant. Theorem \ref{cummutant-fixed pnt} states that the fixed points of a unital channel are the commutant of its Kraus operators, so these will also span the multiplicative domain.
\end{proof}

The next lemma characterizes the projections in $\M_{\UCP^k}$ for higher powers of $k$. Since the multiplicative domain is spanned by its projections, this characterizes the multiplicative domain.
\begin{prop}\label{prop:uebc_md_struc}
	If $\UCP:\M_d\rightarrow \M_d$ is a unital entanglement breaking channel given by $\UCP(x)=\sum_i\varphi_i\zeta_i^*x\zeta_i\varphi_i^*$, then for $1\leq r\leq d$ and any projection $p$, $p\in \M_{\UCP^r}$ if and only if there exists subsets $K_1,\cdots,K_r$ of $[1,\cdots, d]$ such that $\vert K_i\vert=\rm{rank}(p)$ for all $i$, $\sum_{k\in K_1}\zeta_k\zeta_k^*=p$, and $\sum_{k\in K_i}\varphi_k\varphi_k^*=\sum_{k\in K_{i+1}}\zeta_k\zeta_k^*$ for all $1\leq i\leq r-1$.
\end{prop}
\begin{proof}
	To prove one direction, suppose the sets $K_1,\cdots,K_r$ exist. It's clear that for any $K\subseteq[d]$:
\begin{equation}	\label{eq:sum-proj.}
	\UCP\left(\sum_{k\in K}\zeta_k\zeta_k^*\right)=\sum_{i=1}^n\sum_{k\in K}\varphi_i\zeta_i^*\zeta_k\zeta_k^*\zeta_i\varphi_i^*=\sum_{k\in K}\varphi_k\varphi_k^*.
	\end{equation}
	However, we know that $\sum_{k\in K_i}\varphi_k\varphi_k^*=\sum_{k\in K_{i+1}}\zeta_k\zeta_k^*$. Thus, from the Equation \ref{eq:sum-proj.}, $\UCP\left(\sum_{k\in K_i}\zeta_k\zeta_k^*\right)=\sum_{k\in K_{i+1}}\zeta_k\zeta_k^*$, which is still a projection in the multiplicative domain. Since we can repeat this process $r-1$ times, then $\UCP^{r-1}(p)$ is still in $\M_\UCP$; thus $p\in \M_{\UCP^r}$.\par
	For the converse, we will use induction to show that the sets $K_1,\cdots, K_r$ exist. For $r=1$, this reduces to $p=\sum_{k\in K}\zeta_k\zeta_k^*$, which is proven by Proposition \ref{prop:etb_md_struct}. Then suppose the sets exist for all numbers up to $r-1$ and that $p\in \M_{\UCP^r}$ is a projection. Then $p$ must also be in $\M_{\UCP^{r-1}}$, meaning there exist sets $K_1,\cdots, K_{r-1}$ such that $\sum_{k\in K_i}\varphi_k\varphi_k^*=\sum_{k\in K_{i+1}}\zeta_k\zeta_k^*$, for $i\leq r-2$. Now since $p\in \M_{\UCP^r}$, $\UCP^{r-1}(p)\in \M_\UCP$. By the logic in the proof of the previous direction, we have that $\UCP^{r-1}(p)=\sum_{k\in K_{r-1}}\varphi_k\varphi_k^*$. If this is in the multiplicative domain, it must be a  projection of the form $\sum_{k\in K}\zeta_k\zeta_k^*$ for some $K\subseteq[d]$. Thus, $\sum_{k\in K_{r-1}}\varphi_k\varphi_k^*=\sum_{k\in K}\zeta_k\zeta_k^*$; set $K_r=K$ and the conclusion holds.	
\end{proof}
This requirement, where the sums of two different sets of rank-one projections add up to the same projection, is a central part of this construction. Thus, we define it as follows:
\begin{definition}
	A set of vectors $\{\varphi_i\}_{i=1}^n$ in $\C^d$ is \textbf{non-comparable to} another set $\{\zeta_i\}_{i=1}^m$ if there are no two proper subsets $K_1\subsetneq[n-1],K_2\subsetneq[m-1]$ such that $\sum_{k\in K_1}\varphi_k\varphi_k^*=\sum_{k\in K_2}\zeta_k\zeta_k^*$.

	An operator $X:\C^d\rightarrow \C^d$ is said to be non-comparable with respect to a basis $\{\varphi_i\}_{i=1}^d$ if $\{X\varphi_i\}_{i=1}^d$ is non-comparable to $\{\varphi_i\}_{i=1}^d$.
\end{definition}

As an example, the basis:
\[\leftset \frac{1}{\sqrt{2}}\begin{pmatrix}1\\1\\0\\0\end{pmatrix},\frac{1}{\sqrt{2}}\begin{pmatrix}1\\-1\\0\\0\end{pmatrix},
\begin{pmatrix}0\\0\\1\\0\end{pmatrix},\begin{pmatrix}0\\0\\0\\1\end{pmatrix}\rightset\]
is comparable to the basis:
\[\leftset \begin{pmatrix}\tfrac{1}{2}\\\tfrac{\sqrt{3}}{2}\\0\\0\end{pmatrix},\begin{pmatrix}\tfrac{\sqrt{3}}{2}\\-\tfrac{1}{2}\\0\\0\end{pmatrix},
\frac{1}{\sqrt{2}}\begin{pmatrix}0\\0\\1\\1\end{pmatrix},\frac{1}{\sqrt{2}}\begin{pmatrix}0\\0\\-1\\1\end{pmatrix}\rightset\]
even though all of the vectors are distinct. In the notation of Proposition \ref{prop:uebc_md_struc}, $K_1=K_2=\{1,2\}$. \par
As an example of a non-comparable operator, the $d\times d$ normalized discrete Fourier transform matrix $\Delta$, where $\Delta_{ij}=\frac{1}{\sqrt{d}}\omega^{(i-1)(j-1)}$ for some primitive $d$th root of unity $\omega$, is non-comparable to the canonical basis. To see this, the projections $(e_1e_1^*, \cdots, e_de_{d}^*)$ constructed from the canonical basis $(e_1,,\cdots,e_d)$ span the diagonal matrices. However, the projections generated by $\Delta e_i e_i^*\Delta^*$ all have constant diagonals. Thus, the only linear combination of these projections that is diagonal is the identity.  \par 
Using this definition, we can construct entanglement breaking channels on $\M_d$ with any multiplicative index up to $d$.
\begin{theorem}\label{thm:etb_kappa_d}
	For any integer $r$ between $1$ and $d$, there is a unital entanglement-breaking channel $\UCP:\M_d\rightarrow \M_d$ such that $\kappa(\UCP)=r$.
\end{theorem}
\begin{proof}
	Let $\{\zeta_i\}_{i=1}^d$ be any orthonormal basis. For $r=1$, set $\UCP(x)=\sum_{i=1}^d\zeta_i\zeta_i^*x\zeta_i\zeta_i^*$. It's clear that the fixed points $(\rm{Fix}_\UCP)$ of this channel are simply $\Span\{\zeta_i\zeta_i^*|1\leq i\leq d\}$, which is also $\M_\UCP$. Since $\rm{Fix}_\UCP\subseteq\M_{\UCP^\infty}\subseteq\M_\UCP$ in general, we have  $\M_\UCP=\M_{\UCP^\infty}$ and $\kappa(\UCP)=1$.\par
	For $r\geq 2$, let $\{\zeta_i\}_{i=1}^d$ be any orthonormal basis of $\C^d$. Let $U$ be a unitary matrix in $\M_{d-r+2}$ that is non-comparable to $\{\zeta_i\}_{i=1}^{d-r+2}$. Define $\{\varphi_i\}$ such that $\varphi_i=U\zeta_i,i\leq d-r+2$, and $\varphi_i=\zeta_i$ otherwise. Then define
	\begin{align*}
	\UCP(x)=&\sum_{i=1}^d \varphi_i\zeta_{i-1}^*x \zeta_{i-1}\varphi_i^*\\
	=&\sum_{i=2}^{d-r+1}\varphi_i\zeta_{i-1}^*x\zeta_{i-1}\varphi_i^* + \sum_{i=d-r+2}^{d-1}\zeta_i\zeta_{i-1}^*x\zeta_{i-1}\zeta_i^* + \varphi_1\zeta_d^*x\zeta_d\varphi_1^*.
\end{align*}

	(where $\zeta_0=\zeta_d$). In the middle sum, we have replaced $\varphi_i=\zeta_i$, since that was how we constructed $\varphi_i$. From Proposition \ref{prop:etb_md_struct}, we know that the multiplicative domain of $\UCP$ is the span of $\{\zeta_i\zeta_i^*\}_{i=1}^d$. Since $\varphi_i=\zeta_i$ for $i>d-r+2$, we also have that $\Span\{\varphi_i\varphi_i^*\}_{i=d-r+3}^d$ is in the multiplicative domain. Further, since the first $d-r+2$ vectors $\varphi_i$ are non-comparable to $\zeta_i$, we have that for any set $K\subseteq [d-r+2]$ with $\vert K\vert<d-r+2$, $\sum_{k\in K}\varphi_k\varphi_k^*$ is not in the multiplicative domain.\par
	The map $\UCP$ maps the first $d-r+1$ projections in $\M_\UCP$ outside of the multiplicative domain, and then ``pushes" the remaining projections to $\zeta_d\zeta_d^*$, whose image is also outside of the multiplicative domain. \par
	To prove this, take the rank-1 projection $\zeta_{d-r+2}\zeta_{d-r+2}^*$. Note that, by orthogonality, for $d-r+2\leq j\leq d-1$, $\UCP(\zeta_j\zeta_j^*)=\zeta_{j+1}\zeta_{j+1}^*$, which is also in $\M_\UCP$. So applying $\UCP$ $m$ times gives $\UCP^m(\zeta_j\zeta_j^*)=\zeta_{j+m}\zeta_{j+m}^*$ (for $j\leq d-m$). However, suppose $j=d$. Then $\UCP(\zeta_d\zeta_d^*)=\varphi_1\varphi_1^*\notin \M_\UCP$. Thus, for any $j$ with $d-r+2\leq j$, $\UCP^{d-j}(\zeta_j\zeta_j^*)=\zeta_d\zeta_d^*\in \M_\UCP$, but $\UCP^{d-j+1}(\zeta_j\zeta_j^*)\notin \M_\UCP$. Thus, $\zeta_j\zeta_j^*$ is in $\M_{\UCP^{d-j+1}}$ but not $\M_{\UCP^{d-j+2}}$; thus, the multiplicative index must be at least $d-j+2$. Since this pattern works for $j\geq d-r+2$, we get that the multiplicative index must be at least $r$. \par
	To show that this is also an upper bound, let $K$ be an arbitrary subset of $[d]$. We need to show that if $p_K=\sum_{k\in K}\zeta_k\zeta_k^*$ is in $\M_{\UCP^r}$, then $K=[d]$ and $p_K=1$. The idea is as follows: If there is any rank-1 projection missing from $p_K$, then this creates a ``hole". If the hole is in the first $d-r+2$ projections, then the image is not in the multiplicative domain. However, the action of $\UCP$ will move the hole until it is, eventually, in the first $d-r+2$ projections.\par
	Suppose that $j\notin K$ for some $j\in [d]$ (the hole). If $j=d$ or $j<d-r+2$, then $\UCP(p_K)$ is a projection onto a subspace that is (partially) spanned by a strict subset of $\{\varphi_i\varphi_i^*\}_{1\leq i\leq d-r+2}$, which is, by construction, non-comparable to $\{\zeta_i\zeta_i\}$. Thus, $\UCP(p_K)$ cannot be written as a sum of $\zeta_j\zeta_j^*$, and thus is it is not in the multiplicative domain. Now consider that if $j\notin K$ for $p_K$ but $\UCP(p_K)=p_{K'}$ for some $K'$, then $j+1\notin K'$ (since the only way to get $\zeta_{j+1}\zeta_{j+1}^*$ would be if $\zeta_j\zeta_j^*$ was in the decomposition of $P_K$, which it is not). Thus, if $j\geq d$, then $P_K\notin \M_{\UCP^2}$. The pattern continues and thus $p_K\notin \M_{\UCP^m}$ if it is missing any $j\geq d-m+2$. Since we require $P_k\in \M_{\UCP^r}$, we must have that $j\leq d-r+2$. From before, any missing $j$ must be more than $d-r+2$. Thus, to be in $\M_{\UCP^r}$ any missing $j$ must satisfy both $j\leq d-r+2$ and $j>d-r+2$, a contradiction. Thus, there is no $j\notin K$; $K=[d]$ and $p_K=1$.
\end{proof}
As an example in $\M_3$, to get a multiplicative index of 3, we can take the canonical basis $\{e_1,e_2,e_3\}$ and the basis $\{f_1,f_2,f_3\}$ formed by applying the $2\times 2$ discrete Fourier transform matrix to $e_1$ and $e_2$. Let $\omega$ be a primitive 3rd root of unity. Using Theorem \ref{thm:etb_kappa_d}, we can constuct a channel $\UCP$ with Kraus operators of:
\[K_1=f_1e_3^*=\tfrac{1}{\sqrt{2}}\begin{pmatrix}0&0&1\\0&0&1\\0&0&0\end{pmatrix},K_2=f_2e_1^*=\tfrac{1}{\sqrt{2}}\begin{pmatrix}1&0&0\\-1&0&0\\0&0&0\end{pmatrix},K_3=f_3e_2^*=\begin{pmatrix}0&0&0\\0&0&0\\0&1&0\end{pmatrix}.\]
By Proposition \ref{prop:etb_md_struct}, the multiplicative domain of $\UCP$ is all the diagonal matrices, so we take a diagonal matrix $D=\begin{pmatrix}a&0&0\\0&b&0\\0&0&c\end{pmatrix}$. Then
\[\UCP(D)=\frac{1}{2}\begin{pmatrix}a+c&-a+c&0\\-a+c&a+c&0\\0&0&2b\end{pmatrix}.\]
For this to be in $\M_\UCP$, and thus $D\in \M_{\UCP^2}$, the off-diagonal terms must be 0, so $a=c$. If we set $D'$ to be $\begin{pmatrix}a&0&0\\0&b&0\\0&0&a\end{pmatrix}$, then
\[\UCP^2(D')=\frac{1}{2}\begin{pmatrix}a+b&-a+b&0\\-a+b&a+b&0\\0&0&2a\end{pmatrix}\]
and for $D'$ to be in $\M_{\UCP^3}$, $a$ must equal $b$ and thus $D'$ is a multiple of the identity. Since $\M_\UCP\subsetneq \M_{\UCP^2}\subsetneq \M_{\UCP^3}=\C 1$, the multiplicative index of $\UCP$ is 3.
\subsection{Example: Schur Channels}\label{sec:Schur_channels}
Given two matrices $a,b\in \M_d$, the Schur product $a\bullet b$ is defined as the matrix $(a\bullet b)_{ij}=a_{ij}b_{ij}$. A positive semidefinite matrix $b$ whose diagonals are all 1 induces a quantum channel $\mathcal{T}_b$ whose action is given by $\mathcal{T}_b(x)=b\bullet x$. Such channels are necessarily unital. We will call these channels Schur channels.
\par
We will let $b=\begin{pmatrix} J_{d-1}&0\\0&1\end{pmatrix}$, where $J_{d-1}$ is the $d-1\times d-1$ matrix of all ones and let $u$ be the permutation matrix  corresponding to $(1 2 3\cdots d)$ in the symmetric group on $[d]$. 

We can construct a channel $\UCP=\style{U}\circ\mathcal{T}_b$, and this will have multiplicative index $d-1$. Here $\mathcal{U}$ denotes the unitary channel $\mathcal{U}(\cdot)=u(\cdot)u^*$.
To see this, let $p=\tfrac{1}{2}(E_{11}+E_{22}+E_{12}+E_{21})$. This is a projection, and if $k\leq d-2$,
\[\UCP^k(p)=\tfrac{1}{2}(E_{k+1,k+1}+E_{k+2,k+2}+E_{k+1,k+2}+E_{k+2,k+1}).\]
However, $d-1$ applications of $\UCP$ to $p$ gives
\[\UCP^{d-1}(p)=\UCP(\tfrac{1}{2}(E_{d-1,d-1}+E_{dd}+E_{d-1,d}+E_{d,d-1})=\tfrac{1}{2}(E_{11}+E_{dd}),\]
which is no longer a projection. Thus $p\in \M_{\UCP^{d-2}}$ but not $\M_{\UCP^{d-1}}$, so $\kappa(\UCP)\geq d-1$. Since $\M_{\UCP}\cong \M_{d-1}\oplus \M_1$ and $\M_{\UCP^\infty}$ must contain every diagonal matrix, then we can use the same logic as Lemma \ref{lem:chain_length} to see that the chain of multiplicative domains can have length at most $d-1$. Thus, $\kappa(\UCP)=d-1$.


\subsection{Product Channels}\label{sec:product_channels}
Looking at Proposition \ref{prop:kappa_tensor_bound}, the bound on the multiplicative index is linear in the dimension, but cannot increase for product channels. This creates a gap between the possible multiplicative indices of an arbitrary channel and a product channel.
\begin{theorem}\label{thm:index_factor}
	Let $\UCP:\M_d\rightarrow \M_d$ be a unital quantum channel. If $\UCP=\otimes_{i=1}^n \UCP_i$ where $\UCP_i:\M_{d_i}\rightarrow \M_{d_i}$, then $\kappa(\UCP)\leq \max_i\{ 2(d_i-1)\}$. Specifically, if $\kappa(\UCP)\geq d-1$, then $\UCP$ cannot be factored in any way.
\end{theorem}
\begin{proof}
	Suppose $\UCP=\otimes_{i=1}^n\UCP_i$, where $\UCP_i:\M_{d_i}\rightarrow \M_{d_i}$. Then by Proposition \ref{prop:kappa_tensor_bound}
	\[\kappa(\UCP)=\max_i\{\kappa(\UCP_i)\}\leq \max_i\{2(d_i-1)\}.\]\par
	If $n=2$, then $d_1=d/s$ and $d_2=s$ for some $s\vert d$. Since $2\leq s\leq d/2$, then $\max_i\{2(d_i-1)\}\leq 2(d/2-1)=d-2$. Any other factorization will have even smaller components; thus, $\kappa(\UCP)\leq d-2$ for any channel that can be factored. By contrapositive, if $\kappa(\UCP)\geq d-1$, $\UCP$ cannot be factorized.
\end{proof}

Our examples in Sections \ref{sec:ETB_channels} and \ref{sec:Schur_channels} show that it is possible for a channel to be in this gap, and thus we know that neither example can be written as a product of two channels. We summarize this into the following corollary.
\begin{corollary}\label{cor-not-factorable}
For any $d>2$, the channels on $\M_d$ described in Example \ref{sec:Schur_channels} and in Theorem \ref{thm:etb_kappa_d} (with $r=d-1$ or $r=d$) cannot be factored into tensor product of channels acting on smaller subsystems of $\M_d$.
\end{corollary}
\begin{remark}
Note that for a unitary $u\in \M_d\otimes\mathcal{M}_{d'}$ which is not of the form $u_1\otimes u_2$, where $u_1$ is a unitary in $\M_d$ and $u_2$ is a unitary in $\mathcal{M}_{d'}$ respectively, the unitary channel on $\M_d\otimes\mathcal{M}_{d'}$ defined by
\[Ad_{u}(x)=uxu^*, \ \forall x\in \M_d\otimes\mathcal{M}_{d'},\]
is a channel that cannot be factored into product channels (See Example 6.17 in \cite{math-language-qit}). The Swap operator is one such unitary, where the Swap operator ($W$) is defined on the product of two Hilbert spaces by $W(\xi\otimes\eta)=\eta\otimes \xi$. Although for unitary conjugation maps, there are ways (see \cite{Busch}) to decide if they can be factored into two unitary channels or not, for arbitrary channels, it's a hard problem to decide. As the eigenvalues of product channels are all possible product of the constituent channels, analysing  the set of eigenvalues is one such criteria. Even when the eigenvalues match up, Theorem \ref{thm:index_factor} provides a new criteria to decide if a channel can be factored into product channels and Corollary \ref{cor-not-factorable} demonstrates the use of this particular criteria. 
\end{remark}


\subsection{Separable Channels}\label{separable_channels}
While a channel may not be a product of two channels, it could be a convex combination of product channels, known as a separable channel. 
\begin{definition}
	A channel $\UCP:\M_d\otimes \M_c\rightarrow \M_d\otimes \M_c$ is called \textit{separable} if 
	\[\UCP=\sum_{i=1}^n\lambda_i\UCP_{i_1}\otimes \UCP_{i_2},\]
	where $\UCP_{i_1}:\M_d\rightarrow \M_d$ and $\UCP_{i_2}:\M_c\rightarrow \M_c$ are quantum channels, $\lambda_i\geq 0$, and $\sum_{i=1}^n\lambda_i=1$.
\end{definition}
There has been much interest in the separability of states (\cite{math-language-qit}, \cite{Watrous-book}, \cite{Bengtsson}), but less so in the separability of channels (\cite{Watrous-book} Chapter 6). Note that the term ``separable channel" 
is not universal. Some authors prefer to use separable channels just to refer to the sum of product channels; however, we use ``separable channels" to mean convex combination of product channels.   Unfortunately, our technique fails to provide answers on separability, as convex combinations can increase the multiplicative index of a channel. We will use the following lemma to examine this:
\begin{lemma}\label{lem:convex_mult}
	Let $\UCP=\lambda\UCP_1+(1-\lambda)\UCP_2$. Then 
	\[\M_{\UCP^k}=\M_{\UCP_1^k}\cap \M_{\UCP_2^k}\cap \{x\in \M_d|\UCP_1^n(x)=\UCP_2^n(x),1\leq n\leq k\}.\]
\end{lemma}
\begin{proof}
	We will prove this by induction. The result is already established for $k=1$ by Choi (\cite{choi1}, Theorem 3.3) where it was proved that with $\UCP,\UCP_1,\UCP_2$ as above we have 
	\[\M_\UCP=\M_{\UCP_1}\cap\M_{\UCP_2}\cap\{x\in \M_d \ | \ \UCP_1(x)=\UCP_2(x)\}.\] 
	
	so let $k\in\N$ and assume the inductive hypothesis. Then
	\[\UCP^k=\lambda\UCP_1\UCP^{k-1}+(1-\lambda)\UCP_2\UCP^{k-1}.\]
	Again by using Choi's theorem, this makes:
	\[\M_{\UCP^k}=\M_{\UCP_1\UCP^{k-1}}\cap \M_{\UCP_2\UCP^{k-1}}\cap \{x\in \M_d|\UCP_1\UCP^{k-1}(x)=\UCP_2\UCP^{k-1}(x)\}.\]
	Note that by Lemma \ref{lem:md_composition} the first two sets are subsets of $\M_{\UCP^{k-1}}$, so we can assume any $x$ in the last set is in $\M_{\UCP^{k-1}}$. Thus, we can rewrite it as:
	\[\M_{\UCP^k}=\M_{\UCP_1\UCP^{k-1}}\cap \M_{\UCP_2\UCP^{k-1}}\cap \{x\in \M_{\UCP^{k-1}}|\UCP_1\UCP^{k-1}(x)=\UCP_2\UCP^{k-1}(x)\}.\]
	We can write $\UCP^{k-1}=\sum_{i=1}^{2^{k-1}}c_i\UCP_{i_1}\cdots\UCP_{i_{k-1}}$ (where $i_j\in\{1,2\}$ and $\sum_ic_i=1$). If $x\in \M_{\UCP^{k-1}}$ then by the inductive hypothesis this means $\UCP_1^n(x)=\UCP_2^n(x)$, for all $n\leq k-1$. Thus,
	\begin{align*}
	\UCP_{i_1}\cdots\UCP_{i_{k-1}}(x)=&\UCP_{i_1}\cdots\UCP_{i_{k-2}}\UCP_{i_{k-2}}(x)\\
	=&\UCP_{i_1}\cdots\UCP_{i_{k-3}}\UCP_{i_{k-3}}^2(x)\\
	\cdots&\\
	=&\UCP_{i_1}\UCP_{i_1}^{k-2}(x)\\
	=&\UCP_1^{k-1}(x)=\UCP_2^{k-1}(x)
	\end{align*}
	Thus, $\UCP^{k-1}(x)=\sum_{i=1}^{2^{k-1}}c_i\UCP_1^{k-1}(x)=\UCP_1^{k-1}(x)=\UCP_2^{k-1}(x)$. Hence:
	\begin{align*}
	\M_{\UCP^k}=&\M_{\UCP_1\UCP^{k-1}}\cap \M_{\UCP_2\UCP^{k-1}}\cap \{x\in \M_{\UCP^{k-1}}|\UCP_1\UCP_1^{k-1}(x)=\UCP_2\UCP_2^{k-1}(x)\}\\
	=&\M_{\UCP_1\UCP^{k-1}}\cap \M_{\UCP_2\UCP^{k-1}}\cap \M_{\UCP^{k-1}}\cap \{x\in \M_d|\UCP_1^k(x)=\UCP_2^k(x)\}
	\end{align*}
	For the first term, Lemma \ref{lem:md_composition} gives
	\[\M_{\UCP_1\UCP^{k-1}}=\{x\in \M_{\UCP^{k-1}}|\UCP_1(x)\in \M_{\UCP_ 1}\}=\M_{\UCP^{k-1}}\cap \{x\in \M_d|\UCP_1(x)\in \M_{\UCP_1}\}.\]
	Similarly for $\M_{\UCP_2\UCP^{k-1}}$. Combining these results and using the inductive hypothesis for the structure of $\M_{\UCP^{k-1}}$:
	\begin{align*}
	\M_{\UCP^k}=&\{x\in \M_d|\UCP_1(x)\in \M_{\UCP_1}\}\cap \{x\in \M_d|\UCP_2(x)\in \M_{\UCP_2}\}\\
	&\cap \M_{\UCP^{k-1}}\cap \{x\in \M_d|\UCP_1^k(x)=\UCP_2^k(x)\}\\
	=&\{x\in \M_d|\UCP_1(x)\in \M_{\UCP_1}\}\cap \{x\in \M_d|\UCP_2(x)\in \M_{\UCP_2}\}\cap \M_{\UCP_1^{k-1}}\cap \M_{\UCP_2^{k-1}}\\
	&\cap \{x\in \M_d|\UCP_1^n(x)=\UCP_2^n(x), 1\leq n\leq k-1\}\cap \{x\in \M_d|\UCP_1^k(x)=\UCP_2^k(x)\}\\
	=&\leftset x\in \M_{\UCP_1^{k-1}}\midsetl\UCP_1(x)\in \M_{\UCP_1}\rightset\cap \leftset x\in \M_{\UCP_2^{k-1}}\midsetl\UCP_2(x)\in \M_{\UCP_2}\rightset\\
	&\cap\{x\in \M_d|\UCP_1^n(x)=\UCP_2^n(x), 1\leq n\leq k\}\\
	=&\M_{\UCP_1^k}\cap \M_{\UCP_2^k}\cap \{x\in \M_d|\UCP_1^n(x)=\UCP_2^n(x), 1\leq n\leq k\}\numberthis\label{eq:MUCP_convex}
	\end{align*}
\end{proof}
The problem that allows $\kappa(\UCP)>\max\{\kappa(\UCP_1),\kappa(\UCP_1)\}$ is the final set: It might be that the two channels stabilize quickly, but they have different actions on their stabilized multiplicative domains. As an example, consider the two channels $\UCP_1,\UCP_2:\M_3\rightarrow \M_3$ given by:
\[\UCP_1(x)=\frac{1}{3}\begin{pmatrix}1&1&1\\\omega&1&\omega^2\\\omega^2&1&\omega\end{pmatrix}x\begin{pmatrix}1&\omega&\omega^2\\1&1&1\\1&\omega^2&\omega\end{pmatrix},\]
\[\UCP_2(x)=\begin{pmatrix}0&1&0\\0&0&1\\1&0&0\end{pmatrix}x\begin{pmatrix}0&0&1\\1&0&0\\0&0&1\end{pmatrix},\]
where $\omega=e^{i2\pi/3}$. Both of these channels have multiplicative index 1, since they are unitary and $\M_\UCP=\M_{\UCP^\infty}=\M_3$. Let $\UCP=\tfrac{1}{2}\UCP_1+\tfrac{1}{2}\UCP_2$. We calculated that
\[\M_\UCP=\leftset \begin{pmatrix} a&b&-c\\b&2b+a-c&b\\-c&b&a\end{pmatrix}\midsetl a,b,c\in\C\rightset,\, \M_{\UCP^2}=\C 1\]
and thus $\kappa(\UCP)=2$, greater than the multiplicative index of either channel in the convex combination.\par
Extending this idea, we let $\UCP:\M_9\rightarrow \M_9$ be defined as
\[\UCP=\tfrac{1}{2}\UCP_1\otimes \UCP_1 + \tfrac{1}{2}\UCP_2\otimes \UCP_2.\]
This is a separable channel, and both product channels in the convex combination have a multiplicative index of 1 (by Theorem \ref{thm:kappa_bound}) but using numerical methods, we found that $\kappa(\UCP)=2$. Thus, we do not have a direct analog of Theorem \ref{thm:index_factor} for separable channels.\par
However, the size of the multiplicative domain of a convex combination of channels is limited by the size of the multiplicative domain of each channel in the convex combination; thus, for the multiplicative index to be higher than any of the constituent channels, the last set in Equation \ref{eq:MUCP_convex} must be the limiting factor. As a specific case of this, we have the following proposition. For the proof, recall that for three algebras $\A,\B,\mathscr{C}$, if $1\in \B$, then:
\[(1\otimes \A)\cap (\B\otimes\mathscr{C})=1\otimes (\A\cap\mathscr{C}).\]
\begin{prop}
	Let $\{\UCP_1,\cdots,\UCP_n\}$ be unital quantum channels on $\M_d$ such that $\M_{\UCP_i^\infty}=\C 1$ for some $i$ and $\{\Psi_1,\cdots,\Psi_n\}$ be unital channels on $\M_c$. Then if 
	\[\UCP=\sum_{i=1}^n\lambda_i\UCP_i\otimes\Psi_i\]
	where $\lambda_i\geq 0,\sum_i\lambda_i=1$, then $\kappa(\UCP)\leq \max\{2d-2,2c-2\}$.
\end{prop}
\begin{proof}
	Let $k\geq\max\{2d-2,2c-2\}$. Using Lemma \ref{lem:convex_mult} and Theorem \ref{thm:md_tensor}, we have that:
	\begin{align*}
	\M_{\UCP^k}=&\bigcap_{i=1}^n \M_{\UCP_i^k\otimes \Psi_i^k} \cap \leftset x\in \M_{dc}\midsetr \UCP_i^k\otimes\Psi_i^k(x)=\UCP_j^k\otimes \Psi_j^k(x),1\leq i,j\leq n\leq k\rightset\\
	=&\bigcap_{i=1}^n \M_{\UCP_i^\infty}\otimes \M_{\Psi_i^\infty} \cap \leftset x\in \M_{dc}\midsetr \UCP_i^k\otimes\Psi_i^k(x)=\UCP_j^k\otimes \Psi_j^k(x),1\leq i,j\leq n\leq k\rightset
	\end{align*}
	Since $\M_{\UCP_i^\infty}=\C 1$ for some $i$,
	\[\bigcap_{i=1}^n \M_{\UCP_i^\infty}\otimes \M_{\Psi_i^\infty}=\C 1\otimes \bigcap_{i=1}^n \M_{\Psi_i^\infty}.\]
	This gives us:
	\begin{align*}
	\M_{\UCP^k}=&\C 1\otimes \bigcap_{i=1}^n \M_{\Psi_i^\infty}\cap\leftset x\in \M_{dc}\midsetr \UCP_i^k\otimes\Psi_i^k(x)=\UCP_j^k\otimes \Psi_j^k(x),1\leq i,j\leq n\leq k\rightset\\
	=&\C 1\otimes \bigcap_{i=1}^n \M_{\Psi_i^\infty}\cap\leftset 1\otimes x\midsetr x\in \M_c, \Psi_i^k( x)= \Psi_j^k(x),1\leq i,j\leq n\leq k\rightset\\
	=&\C 1\otimes \M_{\Psi^k}
	\end{align*}
	where $\Psi=\sum_{i=1}^n\lambda_i\Psi_i$. Since $\Psi:\M_c\rightarrow \M_c$, then by Theorem \ref{thm:kappa_bound}, $\kappa(\Psi)\leq 2c-2\leq k$, so $\M_{\Psi^k}=\M_{\Psi^\infty}$. Thus:
	\[\M_{\UCP^k}=\C 1\otimes \M_{\Psi^\infty}\]
	for all $k\geq\max\{2d-2,2c-2\}$; hence the multiplicative domain is stable, so $\kappa(\UCP)\leq\max\{2d-2,2c-2\}$.
\end{proof}
The contrapositive states that if $\UCP:\M_d\rightarrow \M_d$ is separable and $\kappa(\UCP)\geq d-1$, then none of the channels in the decomposition of $\UCP$ can have a trivial stabilized multiplicative domain.


\section{Products of Strictly Contractive Channels}\label{Sec:strictly-contractive}
Let $\mathfrak{d}$ be any distinguishability measure 
which is monotone with respect to a quantum channel.
More precisely let $\UCP:\Md\rightarrow\Md$ be a channel such that the following holds 
\[\mathfrak{d}(\UCP(\rho),\UCP(\sigma))\leq\mathfrak{d}(\rho,\sigma),\]
for all density matrices $\rho,\sigma$. A wide class of examples are the quantum f-divergence (\cite{f-div1}, \cite{f-div2}), quantum relative entropy (\cite{Watrous-book}, Theorem 5.38), quantum R{\'e}nyi entropy (\cite{Frank-monotonicity}), fidelity (\cite{Watrous-book}, Theorem 3.30), and trace metric (\cite{math-language-qit}, Proposition 4.37). A channel that does not preserve the measure for any pair of distinct density matrices is called a \textbf{strictly contractive channel}. Formally we define:
\begin{definition}
Given a distinguishability measure $\mathfrak{d}$, a channel $\UCP:\M_d\rightarrow\M_d$ is called a strictly contractive channel with respect to $\mathfrak{d}$ if for every pair of distinct density operators $\rho,\sigma$ in $\M_d$, we have 
\[\mathfrak{d}({\UCP(\rho),\UCP(\sigma)})<\mathfrak{d}(\rho,\sigma)\]
\end{definition}
 Note that in 
(\cite{Raginsky}  and \cite {Doug-Miza}) the authors studied maps with respect to the trace metric and the Bures metric respectively and showed that these maps are dense in the set of all quantum channels.
 Since such channels are found in abundance, we study contractive maps with respect to any distinguishability measure. For these maps we have the following observation:
\begin{lemma}\label{strictly-contractive}
For unital strictly contractive channels $\Ep$ as defined above, the multiplicative domain $\Me$ is trivial, that is, $\Me=\mathbb{C}1$.
\end{lemma} 
\begin{proof}
Suppose $\Me$ is not trivial. Because $1$ is always in $\Me$, we assume that there exists at least one non-trivial element in $\Me$.
Since for unital channels we have the relation  $\Me=\rm{Fix}_{\Ep^*\circ\Ep}$, using the monotonicity property of $\mathfrak{d}$ under $\Ep$ we have that for any $\rho,\sigma\in \Me$, 
\[\mathfrak{d}(\rho,\sigma)=\mathfrak{d}(\Ep^*\circ\Ep(\rho),\Ep^*\circ\Ep(\sigma))\leq\mathfrak{d}(\Ep(\rho),\Ep(\sigma))<\mathfrak{d}(\rho,\sigma).\]
So we arrive at a contradiction and hence $\Me=\mathbb{C}1$.
\end{proof}
As pointed out in (\cite{Raginsky}), it is a hard problem to decide whether the tensor product of strictly contractive channels is strictly contractive or not.
We provide an answer in affirmative if the measure allows the recovery map.

Among all the distinguishability measures, there are certain measures that allow reversibility of a channel. Given a distinguishable measure $\mu$, a channel $\UCP$ is called reversible on a certain subset $\mathcal{S}$ of density operators if there exists a recovery map $\mathcal{R}$ such that $\mathcal{R}\circ\Ep(x)=x$, for all $x\in \mathcal{S}$. This property is often referred to as sufficiency (\cite{Jencova1}, \cite{Jencova2}). There are certain measures which allow reversible maps if the measure of a pair of density matrices is preserved by the channel (for example Renyi entropy, quantum entropy etc. (\cite{f-div2}, \cite{Jencova1}, \cite{Jencova2})).
However there are some measures (trace, fidelity (\cite{Mosonyi}, \cite{Jencova1})) that do not allow the recovery maps even if the channel preserves these measure for some density operators. We denote $\mathfrak{D}_{R}$ to be the measures that allow unital recovery maps if the channel preserves the measure of two density operators. For these measures we have the following theorem:
\begin{theorem}
Let $\mathfrak{d}\in \mathfrak{D}_{R}$. With respect to the measure $\mathfrak{d}$, two unital channels $\UCP_1,\UCP_2$ are strictly contractive if and only if $\UCP_1\otimes\UCP_2$ is also strictly contractive.
\end{theorem}
\begin{proof}
Suppose both $\UCP_1$ and $\UCP_2$ are strictly contractive but $\UCP_1\otimes\UCP_2$ is not. Then there exists a pair of distnct $\rho,\sigma\in \Md\otimes\Md$ such that 
\[\mathfrak{d}(\rho,\sigma)=\mathfrak{d}(\UCP_1\otimes\UCP_2(\rho),\UCP_1\otimes\UCP_2(\sigma)).\] 
Now since $\mathfrak{d}$ allows a recovery map, we get a (unital) channel $\mathcal{R}$ such that 
\[\mathcal{R}\circ(\UCP_1\otimes\UCP_2)(\rho)=\rho \ \text{and} \ \mathcal{R}\circ(\UCP_1\otimes\UCP_2)(\sigma)=\sigma.\]
Now this means (\cite{miza}, Proposition 5.2) we get $\rho,\sigma\in \mathcal{M}_{\UCP_1\otimes\UCP_2}$. Since we get from Theorem \ref{thm:md_tensor} that $\mathcal{M}_{\UCP_1\otimes\UCP_2}=\M_{\UCP_1}\otimes\mathcal{M}_
{\UCP_2}$ and we know by Lemma \ref{strictly-contractive} that $\M_{\UCP_1}=\mathbb{C}1=\mathcal{M}_{\UCP_2}$, we obtain $\mathcal{M}_{\UCP_1\otimes\UCP_2}=\mathbb{C}1$, contradicting the existence of the pair $\rho,\sigma\in \Md\otimes\Md$. 
Hence $\UCP_1\otimes\UCP_2$ is strictly contractive.\par
For the converse, if, without loss of generality, $\UCP_1$ is not strictly contractive, then it admits a recovery map $\mathcal{R}_1$ such that there are two linearly independent matrices $\rho,\sigma$ with $\mathcal{R}_1\circ\UCP_1(\rho)=\rho$ and similarly for $\sigma$. But then the map $\mathcal{R}_1\otimes 1$ will reverse the action of $\UCP_1\otimes\UCP_2$ on the matrices $\rho\otimes 1$ and $\sigma\otimes 1$. This would imply that $\UCP_1\otimes\UCP_2$  cannot be strictly contractive.
\end{proof}


\section{Applications}\label{sec:error_correction}
In the context of quantum information theory, the multiplicative domain of a channel was studied in connection to the scheme of error correction (\cite{choi2}, \cite{johnston}). Specifically, Theorem 11 in \cite{choi2} states that the multiplicative domain is precisely the algebra over the largest unitarily correctable code (UCC). Now suppose we take two channels $\UCP_1:\M_d\rightarrow \M_d$ and $\UCP_2:\M_c\rightarrow \M_c$ and consider the correctable codes of $\UCP_1\otimes \UCP_2$. Clearly, if $\style{C}_1$ and $\style{C}_2$ are unitarily correctable codes for $\UCP_1$ and $\UCP_2$, respectively, then $\style{C}_1\otimes\style{C}_2$ is a correctable code for $\UCP_1\otimes \UCP_2$; however, we would like to have a larger code. Is there a code $\style{C}$ that cannot be decomposed into $\style{C}=\style{C}_1\otimes\style{C}_2$? Using Theorem \ref{thm:md_tensor}, we can show that for unital channels, there are no larger unitarily correctable codes.
\begin{prop}
	Let $\UCP_1:\M_d\rightarrow \M_d$ and $\UCP_2:\M_c\rightarrow \M_c$ be two unital quantum channels. Then $UCC(\UCP_1\otimes\UCP_2)$ is precisely $UCC(\UCP_1)\otimes UCC(\UCP_2)$.
\end{prop}
\begin{proof}
	Using Theorem 11 from \cite{choi2}, $UCC(\UCP_1\otimes\UCP_2)=\M_{\UCP_1\otimes\UCP_2}$, and using Theorem \ref{thm:md_tensor}, this is equal to $\M_{\UCP_1}\otimes \M_{\UCP_2}=UCC(\UCP_1)\otimes UCC(\UCP_2)$.
\end{proof}
For separable channels, the situation does not improve either. If $\UCP=\sum_{i=1}^n\lambda_i\UCP_i\otimes\Psi_i$, then by Choi's theorem (\cite{choi1}, Theorem 3.3), $\M_\UCP\subseteq \M_{\UCP_i}\otimes \M_{\Psi_i}$ for all $i$. Thus, the largest unitarily correctable code of a separable unital channel is at most as large as the tensor product of the largest unitarily correctable codes for any of its constituent product channels.\par
As a natural converse question, if we start with an error correcting code for a channel $\UCP:\M_d\otimes \M_c\rightarrow \M_d\otimes \M_c$, can we correct the code with a unital product channel $\style{R}_1\otimes\style{R}_2$? First, a lemma:
\begin{lemma}\label{lem:ucc_kappa}
If $\UCP$ is a unital channel, $\style{C}$ is the largest $UCC$ (i.e., the multiplicative domain), and $\style{R}$ corrects $\style{C}$, then $\kappa(\style{R})\geq\kappa(\UCP)$.
\end{lemma}
\begin{proof}
Let $k\leq\kappa(\UCP)$. As a first result, $\style{R}^k$ inverts $\UCP^k$ on $\M_{\UCP^k}$. To see this, let $x$ be in $\M_{\UCP^k}$. Then $\UCP^k(x)$ is in the image $\UCP(\M_\UCP)$. Since $\style{R}\circ\UCP$ is the identity on $\M_\UCP$, then $\style{R}\UCP^k(x)=\UCP^{k-1}(x)$. Repeating this argument with each $\style{R}$ gives $\style{R}^k\UCP^k(x)=x$. \par
If $p$ is any projection in $\M_{\UCP^k}$ of rank $r$, then $\UCP^k(p)$ is also a projection of rank $r$. Since $\style{R}^k(\UCP^k(p))=p$, then $\style{R}^k$ takes a rank-$r$ projection to a rank-$r$ projection - in other words, $\UCP^k(p)$ is in $\M_{ \style{R}^k}$. \par
By Theorem 2.5 from \cite{miza}, there is a basis of $\M_{\style{R}^\infty}$ consisting of peripheral eigenoperators of $\style{R}$. Let $x$ be a such a basis element for some eigenvalue $\lambda$. Suppose $x$ is the image $\UCP^k(y)$ for some $k<\kappa(\UCP)$ and some $y\in \M_{\UCP}\setminus \M_{\UCP^\infty}$. Then we would have that $\style{R}^k(x)=\style{R}^k(\UCP^k(y))=y$, but we also have that $\style{R}^k(x)=\lambda^kx$. Thus $\UCP^k(y)=\lambda^{-k}y$ meaning $y\in \M_{\UCP^\infty}$, a contradiction. Thus, $x$ is not in the image of $\UCP^k(\M_\UCP\setminus \M_{\UCP^\infty})$. This means that if we take $x'\in \M_{\UCP^{\kappa-1}}\setminus \M_{\UCP^\infty}$, then $\UCP^{\kappa-1}(x')$ is in $\M_{\style{R}^{\kappa-1}}$ but is not in $\M_{\style{R}^\infty}$. Thus, $\kappa(\style{R})\geq\kappa(\UCP)$.
\end{proof}
As an immediate corollary, we have an extension of Theorem \ref{thm:index_factor} for unital error corrections:
\begin{cor}
If $\style{C}$ is the largest $UCC$ for a unital channel $\UCP:\M_d\otimes \M_c\rightarrow \M_d\otimes \M_c$ and $\kappa(\UCP)\geq \max\{2(d-1),2(c-1)\}$, then $\style{C}$ cannot be corrected by any product channel $\style{R}_1\otimes\style{R}_2$. So if the recovery operator is chosen to be a unitary, then it must be a global unitary channel.
\end{cor}
\begin{proof}
Let $\style{R}$ be any channel that corrects $\style{C}$. From Lemma \ref{lem:ucc_kappa}, $\kappa(\style{R})\geq\max\{2(d-1),2(c-1)\}$, and so from Theorem \ref{thm:index_factor}, $\style{R}$ cannot be written as a product channel. 
\end{proof}
Note that $\UCP^*$ will also correct $\style{C}$ for $\UCP$, and $\UCP$ will correct $\UCP(\style{C})$ for $\UCP^*$, so we have a final corollary:
\begin{cor}
For a unital channel $\UCP$, $\kappa(\UCP)=\kappa(\UCP^*)$.
\end{cor}

\section*{Acknowledgements} 
This work was initiated at the University of Regina under the supervision of Dr. D. Farenick and supported in part by an NSERC Discovery Grant (Farenick) and a University of Regina Graduate Student Research Fellowship.\par
The first author acknowledges the NSERC Canadian Graduate Scholarship and funding from the Department of Combinatorics and Optimization at the University of Waterloo, and the second author acknowledges the Post-Doctoral Fellowship of the Department of Pure Mathematics at the University of Waterloo.\par
We also want to thank our referees, whose suggestions made the proofs of Theorem \ref{thm:any_MUPSA} and Proposition \ref{prop:etb_md_struct} much shorter.
\bibliography{jaques-rahaman}
\bibliographystyle{amsplain}

\end{document}